\documentclass[a4paper,UKenglish,cleveref, autoref, thm-restate]{lipics-v2021}
\usepackage{tablefootnote}
\usepackage[T1]{fontenc}
\usepackage{amsmath,amssymb}
\usepackage{booktabs}
\usepackage{makecell}
\usepackage{capt-of}
\usepackage{wrapfig}
\usepackage{lipsum}
\usepackage{thm-restate}
\usepackage{thmtools}
\usepackage{graphicx}
\usepackage[ruled,vlined,linesnumbered]{algorithm2e}
\usepackage{hyperref}

\hideLIPIcs  

\title{Space and Move-optimal Arbitrary Pattern Formation on Infinite Rectangular Grid by Oblivious Robot Swarm} 

\titlerunning{Arbitrary Pattern Formation on Rectangular Grid by Robot Swarm} 

 \author{Avisek Sharma}{Department of Mathematics, Jadavpur University, India }{aviseks.math.rs@jadavpuruniversity.in}{https://orcid.org/0000-0001-8940-392X}{}

 \author{Satakshi Ghosh}{Department of Mathematics, Jadavpur University, India }{satakshighosh.math.rs@jadavpuruniversity.in}{https://orcid.org/0000-0003-1747-4037}{}
 
 \author{Pritam Goswami}{Department of Mathematics, Jadavpur University, India }{pritamgoswami.math.rs@jadavpuruniversity.in}{https://orcid.org/0000-0001-7008-6135}{}

 \author{Buddhadeb Sau}{Department of Mathematics, Jadavpur University, India }{buddhadeb.sau@jadavpuruniversity.in}{https://orcid.org/0000-0001-7008-6135}{}
 
 \authorrunning{A Sharma et al.} 

 \Copyright{Avisek Sharma, Satakshi Ghosh, Pritam Goswami, and Buddhadeb Sau} 



\category{} 

\relatedversion{} 

\funding{The first and the third authors are supported by the University Grants Commission (UGC), India. The second author is supported by the West Bengal State Government Fellowship Scheme. The fourth author is supported by the Science and Engineering Research Board (SERB), India}

\nolinenumbers 

\EventEditors{John Q. Open and Joan R. Access}
\EventNoEds{2}
\EventLongTitle{42nd Conference on Very Important Topics (CVIT 2016)}
\EventShortTitle{CVIT 2016}
\EventAcronym{CVIT}
\EventYear{2016}
\EventDate{December 24--27, 2016}
\EventLocation{Little Whinging, United Kingdom}
\EventLogo{}
\SeriesVolume{42}
\ArticleNo{23}

\begin{document}

\maketitle

\begin{abstract}
Arbitrary Pattern Formation (APF) is a fundamental coordination problem in swarm robotics. It requires a set of autonomous robots (mobile computing units) to form an arbitrary pattern (given as input) starting from any initial pattern. This problem has been extensively investigated in continuous and discrete scenarios, with this study focusing on the discrete variant. A set of robots is placed on the nodes of an infinite rectangular grid graph embedded in the euclidean plane. The movements of each robot is restricted to one of the four neighboring grid nodes from its current position. The robots are autonomous, anonymous, identical, and homogeneous, and operate Look-Compute-Move cycles. In this work, we adopt the classical $\mathcal{OBLOT}$ robot model, meaning the robots have no persistent memory or explicit communication methods, yet they possess full and unobstructed visibility. This work proposes an algorithm that solves the APF problem in a fully asynchronous scheduler assuming the initial configuration is asymmetric. The considered performance measures of the algorithm are space and number of moves required for the robots. The algorithm is asymptotically move-optimal. Here, we provide a definition of space complexity that takes the visibility issue into consideration. We observe an obvious lower bound $\mathcal{D}$ of the space complexity and show that the proposed algorithm has the space complexity $\mathcal{D}+4$. On comparing with previous related works, we show that this is the first proposed algorithm considering $\mathcal{OBLOT}$ robot model that is asymptotically move-optimal and has the least space complexity which is almost optimal.  
\end{abstract}

\section{Introduction}

Swarm robotics involves a group of simple computing units referred to as \textit{robots} that operate autonomously without having any centralized control. Moreover, the robots are generally anonymous (no unique identifier), homogeneous (all robots execute the same algorithm), and identical (physically indistinguishable). Generally on activation, a robot first takes a snapshot of its surroundings. This phase is called the \textsc{Look} phase. Then based on the snapshot an inbuilt algorithm determines a destination point. This phase is called the \textsc{Compute} phase. Finally, in the \textsc{Move} phase it moves towards the computed destination. These three phases together are called a \textsc{Look-Compute-Move} (LCM) cycle of a robot. 

Through collaborative efforts, these robot swarms can accomplish different tasks such as gathering at a specific point, configuring into predetermined patterns, navigating networks, etc. Presently, the field of robotics research is witnessing significant enthusiasm for swarm robots. The inherent decentralized characteristics of these algorithms provide swarm robots with a notable advantage, as distributed algorithms are both easily scalable and more resilient in the face of errors. Furthermore, swarm robots boast a multitude of real-world applications, including but not limited to tasks like area coverage, patrolling, network maintenance, etc.

In order to accomplish specific tasks, robots require some computational capabilities, which can be determined by various factors such as memory, communication, etc. With respect to memory and communication, the literature identifies two primary robot models. The first one is called the classical $\mathcal{OBLOT}$ model. In this model, the robots are devoid of persistent memory and communication abilities. Another robot model is the $\mathcal{LUMI}$ model where the robots are equipped with a finite number of lights that can take a finite number of different colors. These colors serve as persistent memory (as a robot can see its own color) and communication architecture (as the colors of lights are visible to all other robots). The responsibility for activating robots rests with an entity referred to as the \textit{Scheduler}. Within the existing literature, three primary types of schedulers emerge: \textit{Fully-Synchronous} ($\mathcal{FSYNC}$), \textit{Semi-Synchronous} ($\mathcal{SSYNC}$), and \textit{Asynchronous} ($\mathcal{ASYNC}$).
In the case of fully synchronous and semi-synchronous schedulers, time is partitioned into rounds of uniform length. The duration of the \textsc{Look}, \textsc{Compute}, and \textsc{Move} phases for all activated robots are identical. 
Under a fully-synchronous scheduler, all robots become active at the onset of each round, but in a semi-synchronous setup, not all robots may activate simultaneously in a given round. In an asynchronous scheduler, round divisions are absent. At any given moment, a robot can be either idle or engaged in any of the \textsc{Look}, \textsc{Compute}, or \textsc{Move} phases. The duration of these phases and the spans of robot idleness are finite but unbounded.

The primary focus of this study is to solve the Arbitrary Pattern Formation (APF) problem on an infinite rectangular grid while minimizing spatial utilization. The APF problem involves a group of robots situated within an environment, aiming to create a designated pattern. This pattern is conveyed to each robot as a set of points within a coordinate system as an input. This problem has been extensively studied in the euclidean plane (\cite{BoseDS21,BramasT16,BramasT18,Cicerone19,DieudonnePV10,FlocchiniPSW08,Suzuki96,YAMASHITA10}) and also on a continuous circle \cite{BPAB23}. Bose et al. \cite{BOSE2020} first proposed this problem on a rectangular grid. The rectangular grid is a natural discretization of the plane. To the best of our knowledge, on the discrete domain, this problem has been studied in \cite{BOSE2020,cicerone20,GGSS22,HSVT22,KGGSX22,KGGS2022,SGGS2022}. 
In this paper, the focus is placed on an environment characterized by an infinite rectangular grid. In the upcoming subsection, we delve into the reasons behind the introduction of spatial constraint in the context of this problem.
\subsection{Motivation}

In the majority of previous studies, the implementation of this problem on a grid necessitates a substantial allocation of space (space of a configuration formed by a set of robots is the dimension of the smallest enclosing square of the configuration), even when both the initial and target configurations have minimal spatial requirements. This promptly gives rise to a lot of problems. To begin with, in the scenario where the grid is of bounded dimensions, it is possible that certain patterns cannot be formed, even if robots are initially located within the bounded grid and the target pattern could potentially fit within the grid. This limitation arises due to the existence of intermediate configurations that demand a spatial extent that cannot be accommodated within the confined grid. Moreover, when the spatial demand for an APF algorithm on a grid increases, the count of patterns that can be formed within a bounded grid becomes noticeably fewer compared to the count of patterns formable on the same grid with a lower space requirement. To be more specific, patterns that are ``big enough'' can not be formed if the space requirement is ``big'' on a bounded grid. So, the requirement of large space compromises better utilization of the space.

Moreover, even if complete visibility is entertained for theoretical considerations, this assumption does not hold practical validity within an unbounded environment. In the context of a bounded region, it can be applied with the premise that the environment is finite, and the entire environment falls within the visibility range of each robot. However, introducing the concept of an infinite grid disrupts this assumption. In situations where the grid lacks bounds, it is possible that due to substantial spatial requirements, certain robots might stray beyond the visibility range of others. To the best of our knowledge, there remains an absence of work that addresses the APF challenge within the constraints of limited visibility, an asynchronous scheduler, and the absence of any global coordinate agreement. Thus in this paper, the problem of APF on a grid with minimal spatial requirement has been considered.

\subsection{Related Work} 

In the discrete setting, the problem is first studied in \cite{BOSE2020}. Here, the authors solved the problem deterministically on an infinite rectangular grid with $\mathcal{OBLOT}$ robots in an asynchronous scheduler. Later in \cite{cicerone20}, the authors studied the problem on a regular tessellation graph. In \cite{BOSE2020}, authors count the total required moves asymptotically and also give an asymptotic lower bound for the move complexity, i.e., total number of moves required to solve the problem. In \cite{cicerone20}, authors did not count the total number of moves required for their proposed algorithm. In \cite{GGSS22}, the authors provided two deterministic algorithms for solving the problem in an asynchronous scheduler. The first algorithm of \cite{GGSS22} solves the APF problem for the $\mathcal{OBLOT}$ model. The move complexity of this algorithm matches the asymptotic lower bound given in \cite{BOSE2020}. Thus, this algorithm is asymptotically move-optimal. The second algorithm of \cite{GGSS22} solves the problem for the $\mathcal{LUMI}$ model, and this algorithm is asymptotically move-optimal. Further authors showed that the algorithm is time-optimal, i.e., the number of epochs (a time interval in which each robot activates at least once) to complete the algorithm is asymptotically optimal. In \cite{KGGSX22}, the authors provided a deterministic algorithm for solving the problem with opaque (non-transparent) point robots in the $\mathcal{LUMI}$ model with an asynchronous scheduler assuming one-axis agreement. In \cite{HSVT22}, the authors proposed two randomized algorithms for solving the APF problem in an asynchronous scheduler. The second algorithm works for the $\mathcal{OBLOT}$ model. This algorithm is asymptotically move-optimal and time-optimal. The randomization in this algorithm is only used to break any present symmetry in the initial configuration. If the initial configuration is asymmetric then the algorithm is deterministic. The first algorithm works for opaque point robots with the $\mathcal{LUMI}$ model. This algorithm is also asymptotically move-optimal and time-optimal. In \cite{KGGS2022}, the authors solve the problem with opaque fat robots (robots having nontrivial dimension) with the $\mathcal{LUMI}$ model in an asynchronous scheduler assuming one-axis agreement. In \cite{SGGS2022}, the authors provide an asymptotically move-optimal algorithm solving this problem with robots in the $\mathcal{LUMI}$ model. The work also considered a special requirement and showed that the algorithm is space-optimal. In the next section, we formally state the space complexity of an algorithm and discuss the space complexity of the mentioned works.

\subsection{Space Complexity of APF Algorithms in Rectangular Grid} 
In \cite{SGGS2022}, the authors considered the total space required to execute an algorithm. In Definition~\ref{def1}, we define the space complexity of an algorithm executed by a set of robots on a rectangular grid. Before that let's define the dimension of a rectangle, vertices of which are on some grid nodes, as $m\times n$ if the rectangle has $m$ horizontal grid lines and $n$ vertical grid lines.

\begin{definition}\label{def1}
    The space complexity of an algorithm executed by a set of robots on a rectangular grid is the minimum dimension of the squares (whose sides are parallel with the grid lines) such that no robot steps out of the square throughout the execution of the algorithm.
\end{definition}

Let the smallest enclosing rectangle (SER), the sides of which are parallel to grid lines, of the initial configuration and pattern configuration formed by the robots, respectively, have dimensions $m\times n \ (m\ge n)$ and $m'\times n' \ (m'\ge n')$. Let $\mathcal{D}=\max\{m,n,m',n'\}$. Then the minimum space complexity for an algorithm to solve the APF problem is $\mathcal{D}$. Definition~\ref{def1} assigns a real number to the space complexity that makes it easy to compare different APF algorithms. But consider an APF algorithm that takes a space enclosed by an axis aligned rectangle of dimension $p\times q$. if $M=\max\{m,m'\}$ and $N=\max\{n,n'\}$, then the APF algorithm is better (as far as space is concerned) if $p$ is closer to $M$ and $q$ is closer to $N$.

\paragraph*{Space Complexity of the Previous APF Algorithms}

(\textit{$\mathcal{OBLOT}$ model APF algorithms}) The algorithm proposed in \cite{BOSE2020} has space-complexity at least $2\mathcal{D}$ in the worst case as one of the leaders, named tail moves far away from the rest of the configuration. The first algorithm proposed in \cite{GGSS22} is for the $\mathcal{OBLOT}$ model. It requires the robots to form a compact line. The space complexity of these algorithms is $\mathcal{D}^2$ in the worst case. The second randomized algorithm in \cite{HSVT22} is for the $\mathcal{OBLOT}$ model. In this algorithm, the leader robot moves upwards far away from the rest of the configuration. Thus, it has a space complexity of at least $30\mathcal{D}$ in the worst case.

(\textit{$\mathcal{LUMI}$ model APF algorithms}) The second algorithm proposed in \cite{GGSS22} is for the $\mathcal{LUMI}$ model. This algorithm requires a step-looking configuration where each robot occupies a unique vertical line. Therefore, the space complexity of the algorithm can be $\mathcal{D}^2$ in the worst case. This algorithm needs each robot to have a light with three distinct colors. The first randomized algorithm in \cite{HSVT22} for $\mathcal{LUMI}$ model has space-complexity at least $\mathcal{D}+2$. The authors also did not count the number of lights and colors required for the robots. With a closer look, we observe that this algorithm uses at least 31 distinct colors. Further, deterministic APF algorithms proposed in \cite{KGGSX22,KGGS2022} solved it for obstructed visibility. These works also need the robots to form a compact line, hence the space complexity of these algorithms is $\mathcal{D}^2$ in the worst case. The proposed algorithm in \cite{SGGS2022} has space-complexity $\mathcal{D}+1$ and it requires three distinct colors. 

We say that the first algorithm proposed in \cite{HSVT22} and algorithm proposed in \cite{SGGS2022} are \textit{almost} space-optimal, as the space-complexity is of the form $\mathcal{D}+c$, $\mathcal{D}$ is a lower bound of the space-complexity and $c$ is a constant independent of $\mathcal{D}$. If we consider the rectangle to measure the space, then a rectangle of dimension $M\times N$ is minimally required to solve the APF problem. The first algorithm in \cite{HSVT22} and the algorithm in \cite{SGGS2022} takes space enclosed by rectangle of dimension $(M+2)\times (N+2)$ and $(M+1)\times N$ respectively. We can consider these algorithms as so far the best APF algorithms as far as space complexity is concerned. For the rest of the algorithms one dimension of the rectangle that encloses the required space shoots up twice (algorithm in \cite{BOSE2020}) or 30 times ($2^{nd}$ algorithm in \cite{HSVT22}) or squares (algorithm in \cite{GGSS22,KGGSX22,KGGS2022}). For the rectangle version, if an APF algorithm takes a space of enclosing rectangle of dimension $(M+c_1)\times(N+c_2)$, where $c_1$ and $c_2$ are constants independent of $M$ and $N$, then the algorithm is said to be almost optimal. The challenge of this work is to reconfigure the (oblivious and silent) robots in an optimal space avoiding the occurrence of symmetric configurations and collision among robots while keeping the number of movements asymptotically optimal.

\paragraph*{Our Contribution} First a deterministic algorithm for solving APF in an infinite discrete line is presented. Then exploiting that algorithm this manuscript presents a deterministic algorithm for solving APF in an infinite rectangular grid which is almost space-optimal as well as asymptotically move-optimal. Precisely, the space complexity for the algorithm is $\mathcal{D}+4$ and this algorithm takes a space enclosing the rectangle of dimension $(M+4)\times(N+1)$. The move-complexity of the algorithm is $O(k\mathcal{D})$\footnote{In \cite{HSVT22}, the authors provides this tight lower bound}, where $k$ is the number of robots. The robot model is the classical $\mathcal{OBLOT}$ model and the scheduler is fully asynchronous. To the best of our knowledge so far, this is the first deterministic algorithm solving APF problem in the $\mathcal{OBLOT}$ robot model that has the least space-complexity and optimal move-complexity (See Table~\ref{tab:comp} for comparison with the previous works). The architecture of the description of the algorithm and correctness proof are motivated from \cite{BOSE2020}.

\begin{table}[ht!]
    \centering
    \scriptsize
    \caption{\footnotesize Comparison table}
    \label{tab:comp}
    \begin{tabular}{|p{2cm}|p{1.3cm}|p{2.5cm}|p{2cm}|p{1.5cm}|}
    \hline
        Work & Model & Visibility & Deterministic/ Randomised & Space complexity \\
    \hline
        \cite{BOSE2020} & $\mathcal{OBLOT}$ & Unobstructed & Deterministic& $\ge2\mathcal{D}$\\
    \hline
        $1^{st}$ algorithm in \cite{GGSS22}& $\mathcal{OBLOT}$ & Unobstructed & Deterministic& $\mathcal{D}^2$\\
    \hline
        $2^{nd}$ algorithm in \cite{HSVT22}&  $\mathcal{OBLOT}$ & Unobstructed & Randomised\tablefootnote{The randomisation is only used to break any symmetry present in the initial configuration} & $\ge 30\mathcal{D}$\\
    \hline
        $2^{nd}$ algorithm in \cite{GGSS22}& $\mathcal{LUMI}$ & Unobstructed & Deterministic& $\mathcal{D}^2$\\
    \hline
        $1^{st}$ algorithm in \cite{HSVT22}&  $\mathcal{LUMI}$ & Obstructed & Randomised & $\ge \mathcal{D}$+2\\
    \hline
        \cite{KGGSX22} & $\mathcal{LUMI}$ & Obstructed & Deterministic & $ \mathcal{D}^2$\\
    \hline
        \cite{KGGS2022} & $\mathcal{LUMI}$ & Obstructed (fat robot) & Deterministic & $ \mathcal{D}^2$\\
    \hline
        \cite{SGGS2022} & $\mathcal{LUMI}$ & Unobstructed & Deterministic & $ \mathcal{D}+1$\\
    \hline
        Algorithm in this work & $\mathcal{OBLOT}$ & Unobstructed & Deterministic & $ \mathcal{D}+4$\\ 
    \hline
    \end{tabular}    
\end{table}

\section{Model and Problem Statement}\label{model}

\paragraph*{Robot} 
The robots are assumed to be identical, anonymous, autonomous, and homogeneous. Robots are oblivious, i.e., they do not have any persistent memory to remember previous configurations or past actions. Robots do not have any explicit means of communication with other robots. The robots are modeled as points on an infinite rectangular grid graph embedded on a plane. Initially, robots are positioned on distinct grid nodes. A robot chooses the local coordinate system such that the axes are parallel to the grid lines and the origin is its current position. Robots do not agree on a global coordinate system. The robots do not have a global sense of clockwise direction. A robot can only rest on a grid node. Movements of the robots are restricted to the grid lines, and through a movement, a robot can choose to move to one of its four adjacent grid nodes.

\paragraph*{Look-Compute-Move Cycle.} 
A robot has two states: sleep/idle state and active state. On activation, a robot operates in Look-Compute-Move (LCM) cycles, which consist of three phases. In the Look phase, a robot takes a snapshot of its surroundings and gets the position of all the robots. We assume that the robots have full, unobstructed visibility. In the Compute phase, the robots run an inbuilt algorithm that takes the information obtained in the Look phase and obtains a position. The position can be its own or any of its adjacent grid nodes. In the Move phase, the robot either stays still or moves to the adjacent grid node as determined in the Compute phase.

\paragraph*{Scheduler} The robots work asynchronously. There is no common notion of time for robots. Each robot independently gets activated and executes its LCM cycle. The time length of LCM cycles, Compute phases, and Move phases of robots may be different. Even the length of two LCM cycles for one robot may be different. The gap between two consecutive LCM cycles, or the time length of an LCM cycle for a robot, is finite but can be unpredictably long. We consider the activation time and the time taken to complete an LCM cycle to be determined by an adversary. In a fair adversarial scheduler, a robot gets activated infinitely often.

\paragraph*{Grid Terrain and Configurations} Let $\mathcal{G}$ be an infinite rectangular grid graph embedded on $\mathbb{R}^2$. The $\mathcal{G}$ can be formally defined as a geometric graph embedded on a plane as $\mathcal{P}\times \mathcal{P}$, which is the cartesian product of two infinite (from both ends) path graphs $\mathcal{P}$. Suppose a set of $k>2$ robots is placed on $\mathcal{G}$. Let $f$ be a function from the set of vertices of $\mathcal{G}$ to $\mathbb{N}\cup\{0\}$, where $f(v)$ is the number of robots on the vertex $v$ of $\mathcal{G}$. Then the pair $(\mathcal{G},f)$ is said to be a \textit{configuration} of robots on $\mathcal{G}$. For the initial configuration $(\mathcal{G},f)$, we assume $f(v)\le1$ for all $v$.

\paragraph*{Symmetries} Let $(\mathcal{G},f)$ be a configuration. A \textit{symmetry} of $(\mathcal{G},f)$ is an automorphism $\phi$ of the graph $\mathcal{G}$ such that $f(v)=f(\phi(v))$ for each node $v$ of $\mathcal{G}$. A symmetry $\phi$ of $(\mathcal{G},f)$ is called \textit{trivial} if $\phi$ is an identity map. If there is no non-trivial symmetry of $(\mathcal{G},f)$, then the configuration $(\mathcal{G},f)$ is called an \textit{asymmetric} configuration and otherwise a \textit{symmetric} configuration. Note that any automorphism of $\mathcal{G}=\mathcal{P}\times \mathcal{P}$ can be generated by three types of automorphisms, which are translations, rotations, and reflections. Since there are only a finite number of robots, it can be shown that $(\mathcal{G},f)$ cannot have any translation symmetry. Reflections can be defined by an axis of reflection that can be horizontal, vertical, or diagonal. The angle of rotation can be of $90^{\circ}$ or $180^{\circ}$, and the center of rotation can be a grid node, the midpoint of an edge, or the center of a unit square. We assume the initial configuration to be asymmetric. The necessity of this assumption is discussed after the problem statement.

\paragraph*{Problem Statement}
Suppose a swarm of robots is placed in an infinite rectangle grid such that no two robots are on the same grid node and the configuration formed by the robots is asymmetric. The Arbitrary Pattern Formation (APF) problem asks to design a distributed deterministic algorithm following which the robots autonomously can form any arbitrary but specific (target) pattern, which is provided to the robots as an input, without scaling it. The target pattern is given to the robots as a set of vertices in the grid with respect to a cartesian coordinate system. We assume that the number of vertices in the target pattern is the same as the number of robots present in the configuration. The pattern is considered to be formed if a configuration is formed and that is the same with target pattern up to translations, rotations, and reflections. The algorithm should be \textit{collision-free}, i.e., no two robots should occupy the same node at any time, and two robots must not cross each other through the same edge.  

\paragraph*{Admissible Initial Configurations} We assume that in the initial configuration there is no multiplicity point, i.e., no grid node that is occupied by multiple robots. This assumption is necessary because all robots run the same deterministic algorithm, and two robots located at the same point have the same view. Thus, it is deterministically impossible to separate them afterward. Next, suppose the initial configuration has a reflectional symmetry with no robot on the axis of symmetry or a rotational symmetry with no robot on the point of rotation. Then it can be shown that no deterministic algorithm can form an asymmetric target configuration from this initial configuration. However, if the initial configuration has reflectional symmetry with some robots on the axis of symmetry or rotational symmetry with a robot at the point of rotation, then symmetry may be broken by a specific move of such robots. But making such a move may not be very easy as the robots' moves are restricted to their adjacent grid nodes only. In this work, we assume the initial configuration to be asymmetric.

\section{Space-optimal Arbitrary Pattern Formation on a Grid Line}\label{lineAlgo}
In this section, we solve this problem on a discrete straight line. Suppose we have an infinite path graph $\mathcal{P}=\{(i,i+1)\mid i\in\mathbb{Z}\}$ embedded on a straight line. Suppose $k$ robots are placed on $\mathcal{P}$ at distinct nodes. A configuration is defined similarly as done in the previous section by considering $\mathcal{G}=\mathcal{P}$. The target pattern is given as a set of $k$ distinct positive integers.

\paragraph*{Leader Election and Global Coordinate Setup} We assume the initial configuration of robots does not have reflectional symmetry. First, we set up a global coordinate system that can be agreed upon by all the robots. Suppose $\mathcal{C}$ is a configuration having no reflectional symmetry. For a configuration, we define the smallest enclosing line segment (SEL) to be the smallest line segment in length that contains all the robots in the configuration. Let $\mathcal{L}=AB$ be the SEL of the configuration $\mathcal{C}$. Consider two binary strings of length $|AB|$ (the length of a line segment is the number of grid points on the line segment) called $\lambda_{A}$ and $\lambda_{B}$ with respect to the endpoints of $\mathcal{L}$. Let $\lambda_{A}=\{a_i\}_{i=1}^{|AB|}$ such that $a_i=1$ if and only if the node on the $AB$ line segment having distance $i-1$ from $A$ is occupied by a robot. Similarly, we define $\lambda_{B}$. Since $\mathcal{C}$ has no reflectional symmetry, $\lambda_{A}$ and $\lambda_{B}$ are different. Therefore one of them is lexicographically smaller than the other. Suppose $\lambda_{A}$ is lexicographically smaller than $\lambda_{B}$. Then $A$ is considered as the origin and $\overrightarrow{AB}$ is considered as the positive (right) direction. Also, the robot located at $A$ is said to be \textit{head} and the robot located at $B$ is said to be \textit{tail}. We denote $\mathcal{C}\setminus\{\text{tail}\}$ as $\mathcal{C}'$.

 \paragraph*{Target Embedding} Next we embed the pattern in the following way. Considering the integers given in the target pattern on the number line proceed similarly as done above for $\mathcal{C}$. Let $\mathcal{C}_{target}$ be the target configuration and $A'B'$ be the SEL of $\mathcal{C}_{target}$. Consider two binary strings $\lambda_{A'}$ and $\lambda_{B'}$. If both the strings are equal then the target pattern has a reflectional symmetry. In this case, embed the pattern such that all the target positions are on the right side of the origin except the left most one which is on the origin. If the strings are different then we suppose $\lambda_{A'}$ is the lexicographically smaller one. In this case, embed the pattern such that $A=A'$ and all the target positions are on the right side of the origin. After embedding, the farthest target position from the origin is said to be the \textit{tail-target} and denoted as $t_{target}$. We define, $\mathcal{C}'_{target}=\mathcal{C}_{target}\setminus\{t_{target}\}$.

\paragraph*{Proposed APF algorithm a Line} Next, we describe our proposed algorithm \textsc{ApfLine}. If in a snapshot of a robot, another robot is seen on an edge then the robot discards the snapshot and goes to sleep. Therefore, for simplicity, we assume that any snapshot taken by a robot contains a still configuration $\mathcal{C}$. The head never moves in the algorithm. Firstly, if $\mathcal{C}'=\mathcal{C}'_{target}$ then the tail moves to $t_{target}$. Otherwise, if $t_{target}$ is at the right of the tail, then the tail moves right and the other robots remain static. If $\mathcal{C}'\ne\mathcal{C}'_{target}$, and the tail is at the $t_{target}$ or to the right of the $t_{target}$, then inner robots move to make $\mathcal{C}'=\mathcal{C}'_{target}$. Let $r_i$ be the $i^{th}$ robot from the left and $t_i$ be the $i^{th}$ target position from the left. We try to design the algorithm such that $r_i$ moves to $t_i$. The $r_1$ robot is the head and it is already on $t_1$. If $t_i$ is towards the left of $r_i$ and the left adjacent grid node is empty, then an inner robot $r_i$ moves towards the left. If for each inner robot $r_j$ which is not currently on $t_j$, $t_j$ is at the right of the $r_j$, then an inner robot $r_i$ moves right if $t_i$ is at the right of the $r_i$ and the right adjacent grid node is empty (The pseudo-code of the algorithm is given in Algorithm~\ref{algo:line}).

\begin{algorithm}[ht!]
 \footnotesize
\caption{\textsc{\footnotesize ApfLine} (for a generic robot $r$)}\label{algo:line} 
    \eIf{$\mathcal{C}'=\mathcal{C}'_{target}$}
    {
        tail moves towards $t_{target}$\;
    }
    {
        \eIf{$t_{target}$ is at the right of the tail}
        {
            tail moves towards right\;
        }
        {
            \If{$r=r_i$ is an inner robot}
            {
                \uIf{$t_i$ is at the left of $r_i$}
                {
                    \If{left adjacent grid node is empty}
                    {
                        $r$ moves towards left\;
                    }
                }
                \ElseIf{for each inner robot $r_j$ which is not currently on $t_j$, $t_j$ is at the right of the $r_j$}
                {
                    \If{$t_i$ is at the right of $r_i$}
                    {
                        \If{right adjacent grid node is empty}
                        {
                            $r$ moves towards right\;
                        }
                    }
                }
            }
        }
    }
    \end{algorithm}

\begin{theorem}\label{th0}
    From any asymmetric initial configuration, the algorithm \textsc{ApfLine} can form any target pattern on an infinite grid line within finite time under an asynchronous scheduler.
\end{theorem}
\begin{proof}
    See the Section~\ref{lineProof} of the Appendix.
\end{proof}

\section{The Proposed \textsc{Apf} Algorithm on a Rectangular Grid}\label{gridAlgo}
\subsection{Agreement of a Global Coordinate System and Target Embedding}
Let $\mathcal{C}$ be an asymmetric configuration. Consider the smallest enclosing rectangle (SER) containing all the robots where the sides of the rectangle are parallel to the grid lines. Let $\mathcal{R}=ABCD$ be the SER of the configuration, a $m\times n$ rectangle with $|AD|=m\ge n=|AB|$.
The length of the sides of $\mathcal{R}$ is considered to be the number of grid points on that side. If all the robots are on a grid line, then $R$ is just a line segment. In this case, $\mathcal{R}$ is considered a $m\times 1$ `rectangle' with $A=B$, $D=C$, and $AB=CD=1$.

For a side, say $AB$, of $\mathcal{R}$ we define a binary string, denoted as $\lambda_{AB}$, as follows.
Let $(A=A_1,A_2,\dots,A_m=D)$ be the sequence of grid points on the $AD$ line segment and $(B=B_1,B_2,\dots,B_m=C)$ be the sequence of grid points on the $BC$ line segment. Scan the line segment $AB$ from $A$ to $B$. Then scan the line segments $A_iB_i$ one by one in the increasing order of $i$. The direction of scanning the line segment $A_iB_i$ is set as follows: Scan it from $B_i$ to $A_i$ if $i$ is even and scan it from $A_i$ to $B_i$ if $i$ is odd. While scanning, for each grid point put $0$ or $1$ according to whether it is empty or occupied, respectively (See $\lambda_{AB}$ in Fig.~\ref{fig:lexi}).

\begin{figure}[h]
    \centering
    \includegraphics[width=.22\linewidth]{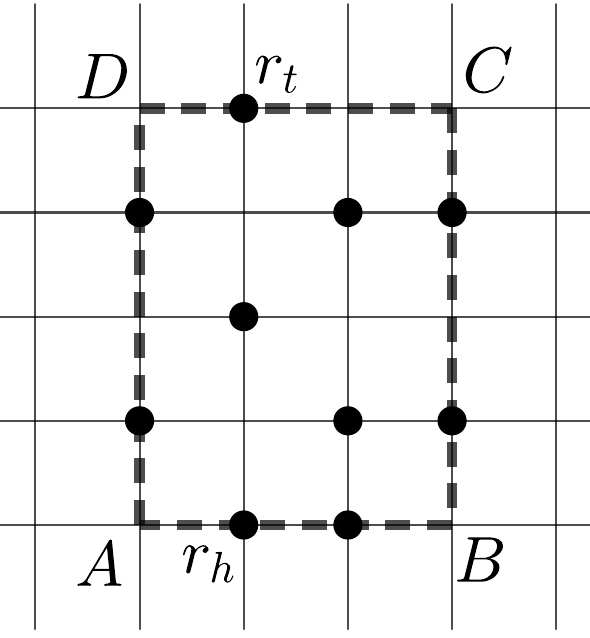}
    \caption{\footnotesize $ABCD$ is the SER of the configuration. $\lambda_{AB}=01101101010011010100$ is the largest lexicographic string, and $r_h$ and $r_t$ are respectively the head and tail robots of the configuration. }
    \label{fig:lexi}
\end{figure}

If $m>n>1$, then for each corner point $A$, $B$, $C$, and $D$, consider the binary strings $\lambda_{AB}$, $\lambda_{BA}$, $\lambda_{CD}$ and $\lambda_{DC}$, respectively. If $m=n>1$, then for each corner point, we have to associate two binary strings with respect to the two sides adjacent to the corner point. Then we have eight binary strings $\lambda_{AB}$, $\lambda_{BA}$, $\lambda_{AD}$, $\lambda_{DA}$, $\lambda_{BC}$, $\lambda_{CB}$, $\lambda_{DC}$ and $\lambda_{CD}$. If any two strings of them are equal then it can be shown that $\mathcal{C}$ has a (reflectional or rotational) symmetry. Since $\mathcal{C}$ is asymmetric, we can find a unique lexicographically largest string (See Fig.~\ref{fig:lexi}). Let $\lambda_{AB}$ be the lexicographically largest string, and then $A$ is considered the \textit{leading corner} of the configuration. The leading corner is taken as the origin, and $\overrightarrow{AB}$ is as the $x$-axis, and $\overrightarrow{AD}$ is as the $y$-axis.

If $\mathcal{R}$ is an $m\times 1$ rectangle, then $\lambda_{AB}$ and $\lambda_{BA}$ are the same string. Then we have two strings to compare. Since the configuration is asymmetric, these two strings must be distinct. Then we shall have a leading corner, say $A=B$. For this case, $A$ is considered as the origin, and $\overrightarrow{AD}$ as the $y$-axis. There will be no agreement of the $x$-axis in this case but since all the robots are on the $y$-axis, so $x$-coordinate of the positions of the robots are 0 at this time.

If $\mathcal{C}$ is asymmetric then a unique string can be elected and hence, all robots can agree on a global coordinate system. By `up' (`down') and `right' (`left'), we shall refer to the positive (`negative') directions of the $x$-axis and $y$-axis of the coordinate system, respectively. The robot responsible for the first 1 in this string is considered the $head$ robot of $\mathcal{C}$ and the robot responsible for the last 1 is considered the $tail$ of $\mathcal{C}$. The robot other than the head and tail is termed the \textit{inner robot}. We define, $\mathcal{C}'=\mathcal{C}\setminus\{\text{tail}\}$ and $\mathcal{C}''=\mathcal{C}\setminus\{\text{head, tail}\}$.

\paragraph*{Target Pattern Embedding} Here we discuss how robots are supposed to embed the target pattern when they agree on a global coordinate system. The target configuration $C_{target}$ is given with respect to some arbitrary coordinate system. Let the $\mathcal{R}'=A'B'C'D'$ be the SER of the target pattern, an $m'\times n'$ rectangle with $|A'D'|\ge |A'B'|>1$. We associate binary strings similarly for $\mathcal{R}'$ as done for $\mathcal{R}$. Let $\lambda_{A'B'}$ be the lexicographically largest (but may not be unique because the $C_{target}$ can be symmetric) among all other strings for $\mathcal{R}'$. The first target position on this string $\lambda_{A'B'}$ is said to be \textit{head-target} and denoted as $h_{target}$ and the last target position is said to be \textit{tail-target} and denoted as $t_{target}$. The rest of the target positions are called \textit{inner target} positions. Then the target pattern is to be embedded such that $A'$ is the origin, $\overrightarrow{A'B'}$ direction is along the positive $x$-axis, and $\overrightarrow{A'D'}$ direction is along the positive $y$-axis. Next, let us consider the case when $|A'B'|=1$, that is when the SER of the target pattern is a line $A'D'$. Let $\lambda_{A'D'}$ be the lexicographically largest string between $\lambda_{A'D'}$ and $\lambda_{D'A'}$. Then the target is embedded in such a way that $A'$ is at the origin and $\overrightarrow{A'D'}$ direction is along the positive $y$-axis. The positive $x$-axis direction can be decided randomly by the robot which first moves out of that line making the SER a rectangle. We define, $\mathcal{C}'_{target}=\mathcal{C}_{target}\setminus\{t_{target}\}$ and $\mathcal{C}''_{target}=\mathcal{C}_{target}\setminus\{h_{target},t_{target}\}$.

\subsection{Outline of the Proposed Algorithm}
    The algorithm is logically divided into seven phases\footnote{The phases are assigned numerical names, yet the sequence of these numerals doesn't precisely correspond to the sequence of their execution during algorithm execution.}. A robot infers which phase it is in from the configuration visible at that time. It does so by checking which conditions in Table~\ref{tab:cond} are fulfilled. We assume that in a visible configuration, no robot is seen on an edge. We maintain such assumption by an additional condition that, if a robot sees a configuration where a robot is on an edge then discard the snapshot and go to sleep. 

    \paragraph*{A Preview of the Algorithm} 
    \begin{itemize}
    
        \item Firstly the tail robot moves upwards to reach a horizontal line such that neither the horizontal line nor other horizontal lines above it contain any robot or target position (Phase~I).
        \item Next the head robot moves left to reach the origin (Phase~II).
        \item Then the tail robot moves a few steps upwards to remove the chance of occurrence of symmetry during the later inner robot movements (phase~I).
        \item Then the tail robot moves rightwards to reach a vertical line such that neither the vertical line nor any vertical line to the right of it contains any robot or target positions (Phase~III).
        \item After that a spanning line is considered (Figure~\ref{fig:line}) and inner robots carefully move along this line (Function \texttt{Rearrange}) to take their respective target position avoiding collision or forming any symmetric configuration (Phase~IV).
        \item After that the tail moves horizontally to reach the vertical line that contains $t_{target}$ (Phase~V).
        \item Then the head robot moves horizontally to reach $h_{target}$ (Phase~VI).
        \item  After that the tail moves vertically to reach $t_{target}$ (Phase~VII).
    \end{itemize}

\begin{table}[ht!]
\caption{\footnotesize Set of conditions on an asymmetric configuration $\mathcal{C}$ having SER $ABCD$ such that the origin is at $A$}
    \label{tab:cond}
    \centering
    \scriptsize
    \begin{tabular}{|c|p{9.5cm}|}
    \hline
        $C_0$ & $\mathcal{C}=\mathcal{C}_{target}$ \\
    \hline
        $C_1$ & $\mathcal{C}'=\mathcal{C}'_{target}$\\
    \hline
        $C_2$ & $\mathcal{C}''=\mathcal{C}''_{target}$\\
    \hline
        $C_3$ & $x$-coordinate of the tail = $x$-coordinate of $t_{target}$\\
    \hline
        $C_4$ & There is neither any robot except the tail nor any target positions on or above $H_t$, where $H_t$ is the horizontal line containing the tail \\
    \hline
        $C_5$ & $y$-coordinate of the tail is odd\\
    \hline
        $C_6$ & SER of $\mathcal{C}$ is not a square\\
    \hline
        $C_7$ & There is neither any robot except the tail nor any target positions on or at the right of $V_t$, where $V_t$ is the vertical line containing the tail\\
    \hline
        $C_8$ & The head is at origin\\
    \hline
        $C_9$ & If the tail and the head are relocated respectively at $C$ and $A$, then the new configuration remains asymmetric\\
    \hline
        $C_{10}$ & $\mathcal{C}'$ has a symmetry with respect to a vertical line\\
    \hline
    \end{tabular}
    
\end{table}

\subsection{Detail Discussion of the Phases}
\paragraph*{Phase I} A robot infers itself in Phase~I if $\neg(C_4\land C_5\land C_6)\land\neg(C_1\land C_3)$ is true. In this phase, the tail moves upward and all other robots remain static. The aim of this phase is to make $C_4\land C_5\land C_6$ true.

\paragraph*{Phase II} A robot infers itself in Phase~II if $(C_4\land C_5\land C_6\land \neg C_8)\land((C_2\land\neg C_3)\lor\neg C_2)$ is true. In this phase, the head moves towards the left, and other robots remain static. This phase aims to make $C_8$ true.

\paragraph*{Phase III} A robot infers itself in Phase~III if $C_4\land C_5\land C_6 \land C_8 \land \neg C_2 \land \neg C_7$ is true. The aim of this phase is to make $C_7$ true. In this phase, there are two cases to consider. The robots will check whether $C_{10}$ is true or not. If $C_{10}$ is false, then robots check whether $C_9$ is true or not. If $C_9$ is not true then the tail moves upward. Otherwise, the tail moves right or upwards in accordance with $m>n+1$ or $m=n+1$ (dimension of the current SER is $m\times n$ with $m\ge n$). If $C_{10}$ is true, then the tail moves left or upwards in accordance with $m>n+1$ or $m=n+1$. Other robots remain static in both cases.

\paragraph*{Phase IV} A robot infers itself in Phase~IV if $C_4\land C_5\land C_6 \land C_7 \land C_8 \land \neg C_2 $ is true. In this phase, the inner robots execute function~\texttt{Rearrange} to make $C_2$ true.

\paragraph*{Function \texttt{Rearrange}} In this function inner robots move to take their respective target positions. Let $\mathcal{C}$ be the current configuration. Let $ABCD$ be the SER of $\mathcal{C}$. According to the assumption exactly two nonadjacent vertices are occupied by robots in rectangle $ABCD$. Specifically, these two robots are the head and the tail of the configuration. Let the head and tail be located at $A$ and $C$ respectively. Consider the path $\mathcal{P}$ starting from $A$ to $C$ as illustrated in bold edges in Fig.~\ref{fig:line}. Inner robots adopt algorithm \textsc{ApfLine} considering this path as the line. Here, we define a robot $r'$ at the \texttt{left} (\texttt{right}) side of $r$ if $r'$ is closer to the head (tail) than $r$ in $\mathcal{P}$. Let us order the target positions. Denote $h_{target}$ as $t_1$, then the next closest target position from the head in $\mathcal{P}$ as $t_2$. Similarly, denote the $i^{th}$ closest target positions in $\mathcal{P}$ from the head as $t_i$. Note that, $t_k$ is the $t_{target}$. Similarly order all the robots, $\{r_i\}_{i=1}^k$, where $r_1$ is the head and $r_i\ (i>1)$ is the $i^{th}$ closest robot from the head on $\mathcal{P}$.

\begin{figure}[h!]
    \centering
    \includegraphics[width=.3\linewidth]{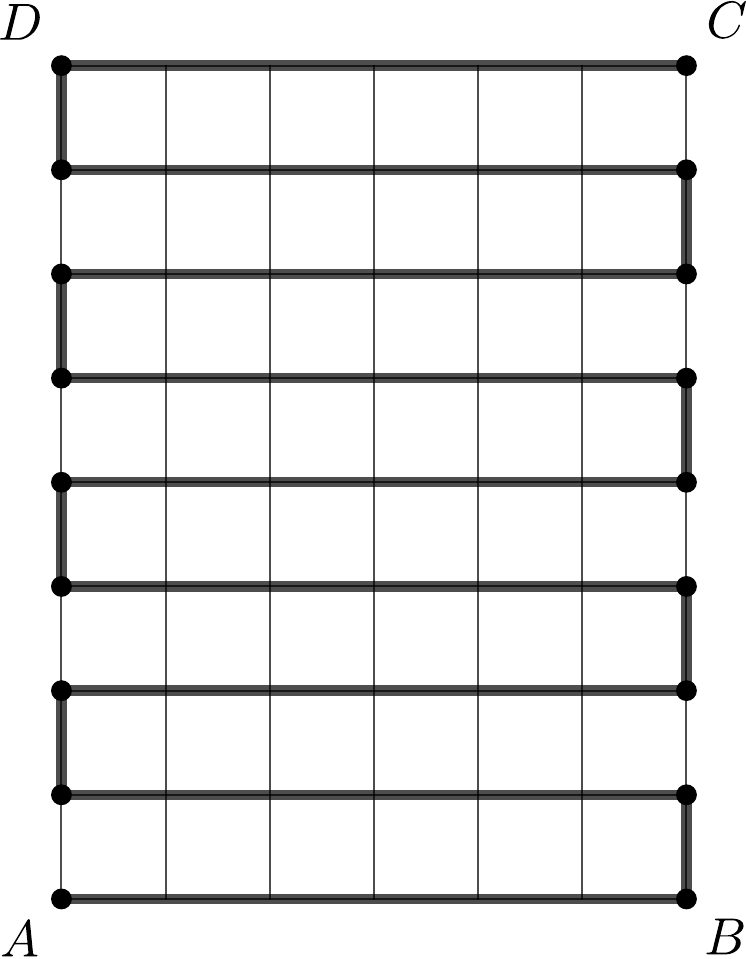}
    \caption{\footnotesize Path joining the nodes $A$ and $C$ mentioned in bold edges}
        \label{fig:line}
\end{figure}

If $t_i$ is at the \texttt{left} of $r_i$ and there are no other robots in the sub-path of $\mathcal{P}$ starting from the position of $r_i$ to $t_i$, then $r_i$ moves to $t_i$. The movement strategy is described as follows. If $r_i$ and $t_i$ are at the same vertical (or, horizontal) line then $r_i$ moves through the vertical (or, horizontal) line joining $r_i$ and $t_i$. Suppose, $r_i$ and $t_i$ are not at the same vertical line or horizontal line. If the downward adjacent vertex of $r_i$ is at the \texttt{right} of $t_i$ then $r_i$ moves downwards. If the downward adjacent vertex is at the \texttt{left} of $t_i$, then $r_i$ moves to its \texttt{left} adjacent node on $\mathcal{P}$.

If there is no robot $r_j$ such that $t_j$ is at the \texttt{left} of $r_j$, then movements of an inner robot towards \texttt{right} start. If $t_i$ is at the \texttt{right} of $r_i$, and there are no other robots in the sub-path of $\mathcal{P}$ starting from the position of $r_i$ to $t_i$, then $r_i$ moves to $t_i$. The movement strategy is described as follows. If $r_i$ and $t_i$ are at the same vertical (or, horizontal) line then $r_i$ moves through the vertical (or, horizontal) line joining $r_i$ and $t_i$. Suppose, $r_i$ and $t_i$ are not at the same vertical line or horizontal line. If the upward adjacent vertex of $r_i$ is at the \texttt{left} of $t_i$ then $r_i$ moves upwards. If the upward adjacent vertex is at the \texttt{right} of $t_i$, then $r_i$ moves to its \texttt{right} adjacent on node $\mathcal{P}$ (pseudo code of the function~\texttt{Rearrange} is given Algorithm~\ref{rearrange}).

\begin{algorithm}[ht!]
\footnotesize
    \caption{\footnotesize Function \texttt{Rearrange} for a robot $r=r_i$}
    \label{rearrange}
    \uIf{$t_i$ is at the \texttt{left} of $r_i$}
    {
        \If{there are no other robot in the sub-path of $\mathcal{P}$ starting from position of $r_i$ to $t_i$}
        {
              \eIf{$r_i$ and $t_i$ are at the same vertical (or, horizontal) line}
              {
                    $r_i$ moves towards $t_i$ through the vertical (or, horizontal) line joining $r_i$ and $t_i$\;
              }
              {
                    \eIf{the downward adjacent vertex of $r_i$ is at the \texttt{right} of $t_i$}
                    {
                        $r_i$ moves downwards\;
                    }
                    {
                        $r_i$ moves to its \texttt{left} adjacent node on $\mathcal{P}$\;
                    }
              }
        }
    }
    \ElseIf{$t_i$ is at the \texttt{right} of $r_i$}
    {
        \If{there is no inner robot $r_j$ such that $t_j$ is at the \texttt{left} of $r_j$}
        {
            \If{there are no other robot in the sub-path of $\mathcal{P}$ starting from position of $r_i$ to $t_i$}
            {
                \eIf{$r_i$ and $t_i$ are at the same vertical (or, horizontal) line}
                {
                    $r_i$ moves towards $t_i$ through the vertical (or, horizontal) line joining $r_i$ and $t_i$\; 
                }
                {
                    \eIf{the upwards adjacent vertex of $r_i$ is at the \texttt{left} of $t_i$}
                    {
                        $r_i$ moves upwards\;
                    }
                    {
                        $r_i$ moves to its \texttt{right} adjacent node on $\mathcal{P}$\;
                    }
                }
            }
        }
    }
\end{algorithm}

\paragraph*{Phase V} A robot infers itself in Phase~V if $C_2 \land C_4\land C_5 \land C_6 \land C_8 \land \neg C_3$ is true. In this phase, the tail moves horizontally to make $C_3$ true. Let $H_t$ be the horizontal line containing the tail and $T'$ be the point on the $H_t$ that has the same $x$-coordinate with $t_{target}$. If $C_{10}$ is not true then the tail moves horizontally towards $T'$. Next let $C_{10}$ be true. Let $ABCD$ be the SER of the current configuration $\mathcal{C}$ and $AB'C'D'$ be the SER of $\mathcal{C}'$. Let $C''$ be the point where line $B'C'$ intersects with $H_t$. Let $E$ be the point on the $H_t$ (See Figure~\ref{fig:phV}). Let the tail robot be at $T$. If both $T$ and $T'$ are at the right side of $C''$ or in on the line segment $DE$, then the tail moves towards $T'$. Otherwise, the tail moves leftward.

\begin{figure}[h!]
    \centering
    \includegraphics[width=4.5cm]{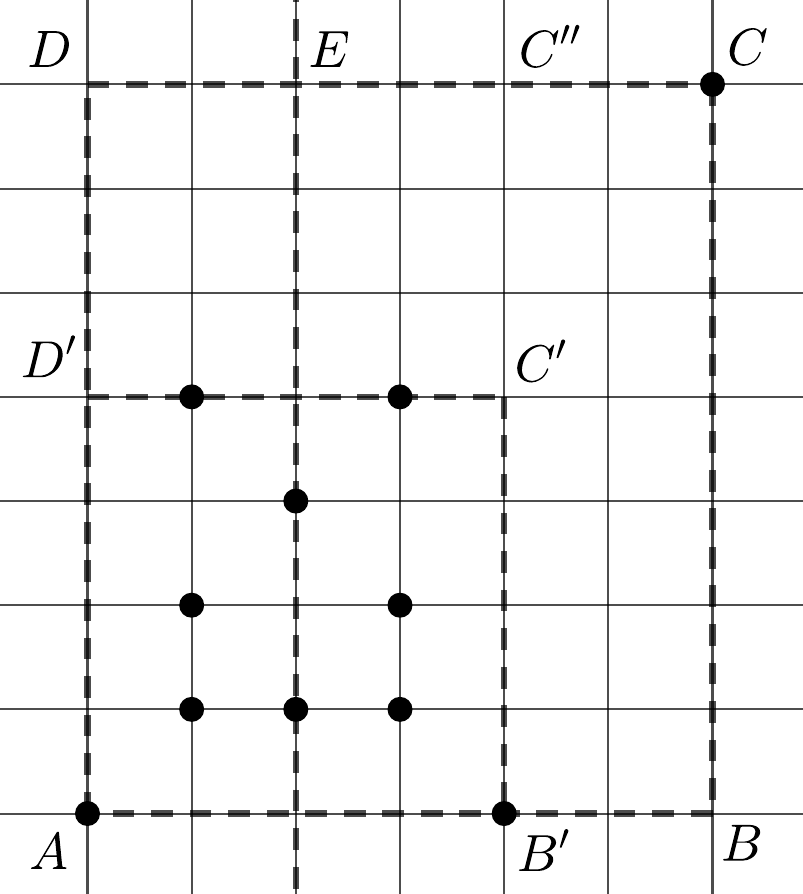}
    \caption{An image related to Phase~V}
    \label{fig:phV}
\end{figure}

\paragraph*{Phase VI} A robot infers itself in Phase~VI if $\neg C_1 \land C_2 \land C_3\land C_4\land C_5\land C_6$ is true. In this phase, the head moves horizontally to reach $h_{target}$. After the completion of this phase, $\neg C_0\land C_1 \land C_3$ becomes true.

\paragraph*{Phase VII} A robot infers itself in Phase~VII if $\neg C_0\land C_1 \land C_3$ is true. In this phase, the tail moves vertically to reach $t_{target}$.

\subsection{Correctness and Performance of the Proposed Algorithm}
In this section, we prove the correctness of the proposed algorithm. First, we show (in Lemma~\ref{phase0} in Appendix) that any initial asymmetric configuration for which $C_0$ is not true falls in one of the seven phases. Then we show that from any asymmetric initial configuration, the algorithm allows the configuration to satisfy $C_0=$ true after passing through several phases.

\begin{theorem}\label{main}
    The proposed algorithm can form any pattern consisting of $k$ points by a set of $k$ oblivious asynchronous robots if the initial configuration formed by the robots is an asymmetric configuration and has no multiplicity point. 
\end{theorem}

    Recall the Definition~\ref{def1} of the space complexity of an algorithm executed by a set of robots on an infinite rectangular grid. 
    In Theorem~\ref{space}, we calculate the space complexity of the proposed algorithm. The move complexity is recorded in the Theorem~\ref{th:move}.

    \begin{theorem}\label{space}
     Let $\mathcal{D}=\max\{m,n,m',n'\}$ where $m\times n$ $(m\ge n)$ is the dimension of the SER of the initial configuration and $m'\times n'$ $(m'\ge n')$ is the dimension of the SER of the target configuration. Then the space complexity of the proposed algorithm is at most $\mathcal{D}+4$. More precisely, if $M=\max\{m,m'\}$ and $N=\max\{n,n'\}$, then the proposed algorithm takes the space enclosed by a rectangle of dimension $(M+4)\times (N+1)$.
    \end{theorem}

    \begin{theorem}\label{th:move}
    The proposed algorithm requires each robot to make $O(\mathcal{D})$ movements, hence the move-complexity of the proposed algorithm is $O(k\mathcal{D})$.
    \end{theorem}

   \section{Conclusion}\label{concl}
    This work first provides an algorithm that solves the APF problem in an infinite line by a robot swarm. Then adopting the method, it provides another algorithm that solves the APF in an infinite rectangular grid by a robot swarm. The robots are autonomous, anonymous, identical, and homogeneous. The robot model used here is the classical $\mathcal{OBLOT}$ model. The robots work under a fully asynchronous scheduler. The proposed algorithm is almost space-optimal (Theorem~\ref{space}) and asymptotically move-optimal (Theorem~\ref{th:move}).

    A few limitations of this work are the following. Here we assume that the initial configuration is asymmetric. Finding complete characterization of the initial configurations from which APF can be solved deterministically is an interesting future direction.    Next, the version of the APF problem under consideration does not permit multiple points in the target configuration. More precisely, the number of target positions in the target pattern is equal to the number of robots within the system. Solving a more generalized version of the problem that allows target patterns with target positions less than the total number of robots, is a possible future direction. Next, the proposed algorithm is almost space optimal, so finding out the exact lower bound when starting from an asymmetric initial configuration is an interesting direction. Also, this does not consider time-optimality, so considering all the three parameters space, move and time at the same time can be an interesting future work.

\bibliography{lipics-v2021-sample-article}

\newpage
\begin{center}
    \Large\textbf{Appendix}
\end{center}

\section{Correctness of the Algorithm \textsc{ApfLine}}\label{lineProof} 

For two finite binary strings $\lambda_1$ and $\lambda_2$ of same length, $\lambda_1\prec (\preceq)\lambda_2$ shall denote that $\lambda_1$ is lexicographically less than (less or equal to) $\lambda_2$.
Next, we make a simple observation in Proposition~\ref{prop1} and then another observation in Proposition~\ref{prop2}. 
\begin{proposition}\label{prop1}
 Let $\mathcal{C}$ be a still configuration and $AB$ the SEL of $\mathcal{C}$. Let $\lambda_{A}$ and $\lambda_{B}$ be the binary strings for the two endpoints. Suppose an inner robot, say $r$, moves to its adjacent empty grid towards $A$. Let $\lambda^{new}_{A}$ and $\lambda^{new}_{B}$ be the new binary strings after the movement, then $\lambda_{A}\prec\lambda^{new}_{A}$ and $\lambda^{new}_{B}\prec\lambda_{B}$.  
\end{proposition}

\begin{proposition}\label{prop2}
    Let $\mathcal{C}$ be a configuration and $\mathcal{L}=AB$, the SEL of $\mathcal{C}$ such that $\lambda_{B}\preceq \lambda_{A}$. Let the robot situated at $B$ move outside the $\mathcal{L}$ and reach a point $B'$, and the new configuration become $\mathcal{C}'$. If $\lambda_{A}'$ and $\lambda_{B'}'$ are the binary strings for $\mathcal{C}'$, then $\lambda_{B'}'\prec\lambda_{A}'$.

\end{proposition}
    
\begin{proof}
    Let $\lambda_{A}=a_1a_2\dots a_{n-1}a_n$ and $\lambda_{B}=b_1b_2\dots b_{n-1}b_n$. Then according to the configuration $\mathcal{C}'$,  $\lambda'_{A}=a_1a_2\dots a_{n-1}0a_n$ and $\lambda'_{B'}=b_10b_2\dots b_{n-1}b_n$.
    Since the total number of robots present in the system is greater than 2, so there exists the smallest $i$ such that $2\le i\le n-1$ and $a_i=1$. If $i=2$, that is, $a_2=1$, then $\lambda_{B'}'\prec\lambda_{A}'$ because the second entry of $\lambda_{B'}'$ is zero. If $i>2$, then $a_{i-1}=b_{i-1}=0$ is the $i^{th}$ entry of $\lambda_{B'}'$ whereas $i^{th}$ entry of $\lambda_{A}'$ is $a_i=1$. The first $i-1$  corresponding entries of $\lambda_{B'}'$ and $\lambda_{A}'$ are equal. Thus, $\lambda_{B'}'\prec\lambda_{A}'$.
    
    Next let $\lambda_{B}\prec\lambda_{A}$. For this case, let $j^{th}$ entry is the first entry from left where $\lambda_A$ and $\lambda_B$ differ. Then $a_j=1$ but $b_j=0 $, $a_i=b_i$ for all $i<j$. If $a_2=1$, then $\lambda_{B'}'\prec\lambda_{A}'$. Otherwise, if $a_2=0$ then from $\lambda_{B}\prec\lambda_{A}$, we have $b_2$ must be zero. Again if $a_3=1$, $\lambda_{B'}'\prec\lambda_{A}'$. Otherwise, if $a_3=0$ then $b_3$ must be zero from $\lambda_{B}\prec\lambda_{A}$. Similarly proceeding, we get $a_2=b_2=\dots a_{j-1}=b_{j-1}=0$. Then $j^{th}$ entry of $\lambda_{A}'$ is 1 but $j^{th}$ entry of $\lambda_{B'}'$ is $b_{j-1}=0$. Thus, $\lambda_{B'}'\prec\lambda_{A}'$.
\end{proof}

Next, in the Algorithm~\ref{algo:line}, there are so-called four types of movements. These four types are respectively in lines 2, 5, 10, and 14 of the Algorithm~\ref{algo:line}. The first two types are the movements by the tail and the other two types are movements by inner robots. In Lemma~\ref{l1} and Lemma~\ref{l2}, it is proved that through the movements of inner robots, the configuration remains asymmetric and the coordinate system does not change until the target pattern is formed. In Lemma~\ref{l3}, it is shown that after finite movements of inner robots, all the inner robots occupy their respective target positions. In Lemma~\ref{tail1} and Lemma~\ref{tail2}, it is shown that the movements by tails do not bring any symmetry in the configuration and the coordinate system remains unchanged until the target pattern is formed. Also, it is shown that, after a finite time the prescribed goal is achieved via the movements of the tail. Finally, in Theorem~\ref{th0} the correctness of the Algorithm~\ref{algo:line} is proved. 

\begin{lemma}\label{l1}
  If an inner robot moves left then the configuration remains asymmetric and the coordinate system does not change.  
\end{lemma}
\begin{proof}
   Let at time $t$, the head and tail robot is at $A$ and $B$, respectively. Then $\lambda_{B}\prec\lambda_{A}$. Suppose, an inner robot moves towards the left and $\lambda^{new}_{A}$ and $\lambda^{new}_{B}$ be the updated respective binary strings. Then from Proposition~\ref{prop1}, $\lambda_{A}\prec\lambda^{new}_{A}$ and $\lambda^{new}_{B}\prec\lambda_{B}$. Therefore, combining the three inequalities we get $\lambda^{new}_{B}\prec\lambda^{new}_{A}$. Thus, the new configuration is asymmetric, and the coordinate system does not change. 
\end{proof}

\begin{lemma}\label{l2}
    If an inner robot moves rightwards according to the Algorithm~\ref{algo:line}, then the configuration remains asymmetric, and the coordinate system does not change unless the target pattern has been formed. 
\end{lemma}
\begin{proof}
    At time $t$, suppose the head and tail robots is at $A$ and $B$, respectively. Then $\lambda_{B}\prec\lambda_{A}$. Suppose, an inner robot $r_k$ is about to move towards right according to the Algorithm~\ref{algo:line} and after the movement $\lambda^{new}_{A}$ and $\lambda^{new}_{B}$ be the updated respective binary strings. According to the Algorithm~\ref{algo:line} the scenario is, some inner robots are on their respective target positions, and the respective target positions of the rest inner robots are on their right side. The robot $r_k$ is one of them. Also, at this time the tail is either at the $t_{target}$ or at the right of $t_{target}$. Let's consider a configuration $\mathcal{C}'$ considering the tail at $B$ and all other target positions but $t_{target}$ are occupied by robots. Let $\lambda'_A$ and $\lambda'_B$ be the binary strings for the end points $A$ and $B$, respectively, in $\mathcal{C}'$. If $t_{target}$ is at $B$ then according to the target embedding $\lambda'_B\preceq\lambda'_A$. Otherwise, if $B$ is at the right of the $t_{target}$, from Proposition~\ref{prop2} $\lambda'_B\prec\lambda'_A$. Thus combining both, we can say $\lambda'_B\preceq\lambda'_A\dots(1)$.
    
    The proof would be done if we show that $\lambda^{new}_{B}\prec\lambda^{new}_{A}$. On the contrary, if possible let $\lambda^{new}_{A}\preceq\lambda^{new}_{B}\dots(2)$. Since the target pattern has not been formed, some inner robots still need to move right. Thus, from Proposition~\ref{prop1} we have $\lambda^{new}_{B}\prec\lambda_B'\dots(3)$ and $\lambda'_{A}\prec\lambda_A^{new}\dots(4)$. Then, from $(2)$ and $(3)$, we have $\lambda_A^{new}\prec\lambda'_B\dots(5)$. From $(4)$ and $(5)$, we have $\lambda_A'\prec\lambda_B'$, which contradicts $(1)$. 
\end{proof}

\begin{lemma}\label{l3}
 If at some time $t$, we have an asymmetric configuration $\mathcal{C}$ such that $\mathcal{C}'\ne\mathcal{C}'_{target}$, and the tail is at the $t_{target}$ or at the right of the $t_{target}$, then after a finite time, through the movements of inner robots, $\mathcal{C}'=\mathcal{C}'_{target}$ becomes true.
\end{lemma}
\begin{proof}
    If the inner robots successfully can move according to our algorithm then after a finite time all inner robots take their respective target positions making $\mathcal{C}'=\mathcal{C}'_{target}$ true. The matters that need to be taken care of are (1)~the head and the tail robot remain head and tail, respectively (hence, the coordinate system remains unchanged), and (2)~no collision or deadlock occurs throughout the movement of the inner robots. From Lemma~\ref{l1} and lemma~\ref{l2}, we have that throughout the execution of the algorithm~\ref{algo:line} when the inner robots move, the head and tail remains head and tail respectively and hence the coordinate system remains unchanged. Next, According to the Algorithm~\ref{algo:line} an inner robot only moves to its empty adjacent grid node. Thus, a collision can only happen if two robots move to the same node. Suppose, two inner robots $r_i$ and $r_{i+1}$ move to an empty node $v$. This implies, either $t_i$ and $t_{i+1}$ both are at $v$ or $t_i$ is at the right of $t_{i+1}$, which is not possible. Thus, no collision will occur. Then we show no deadlock will be created throughout the movement of the inner robots. A deadlock can occur if there are two inner robots $r_i$ and $r_{i+1}$ adjacent to each other, where $r_i$ wants to move towards right and $r_{i+1}$ is either at $t_{j+1}$ or wants to move towards left. From a similar argument as before, this is not possible.

\end{proof}

\begin{lemma}\label{tail1}
  If at some time $t$, we have an asymmetric configuration $\mathcal{C}$ such that $\mathcal{C}'\ne\mathcal{C}'_{target}$, and the tail is at the left of $t_{target}$, then after a finite number of moves by tail towards right the tail reaches $t_{target}$. 
\end{lemma}
\begin{proof}
    Let $AB$ be the SEL of the $\mathcal{C}$ such that the tail is at $B$. Let $\lambda_{A}$ and $\lambda_{B}$ be the binary strings in $\mathcal{C}$. Then $\lambda_{B}\prec\lambda_{A}$. Suppose the tail moves towards the right and moves to a point $B'$, and the new configuration is $\mathcal{C}'$. Let $\lambda_A'$ and $\lambda_{B'}'$ be the binary strings from the endpoints of the SEL of $\mathcal{C}'$. Then from Proposition~\ref{prop2}, $\lambda_{B'}'\prec \lambda_A'$. Thus, after the movement of the tail towards the right, the configuration remains asymmetric and the tail robot remains the tail. Thus, after a finite move towards $t_{target}$, the tail reaches $t_{target}$.\hfill $\Box$
\end{proof}
    
\begin{lemma}\label{tail2}
    If at some time $t$, the configuration is $\mathcal{C}$ such that $\mathcal{C}'=\mathcal{C}'_{target}$ is true, then after a finite move of the tail towards $t_{target}$, the target pattern gets formed.
\end{lemma}
\begin{proof}
    
    Let $AB$ be the SEL of the configuration $\mathcal{C}$ such that the tail is situated at $B$. Suppose the embedding of the target $t_{target}$ is at $B_t$. If $B_t$ and $B$ are one hop away from each other then after one movement of the tail target pattern will be formed. Suppose $B_t$ and $B$ are not adjacent to each other. According to the Algorithm~\ref{algo:line} the tail moves towards $B_t$. Suppose after one movement the tail lands on the point $B'$ and the configuration becomes $\mathcal{C}'$. We need to show the tail remains the tail after moving to $B'$. Let $\lambda_A'$ and $\lambda_{B'}'$ be the binary strings from the endpoints of the SEL of $\mathcal{C}'$. There are two exhaustive cases: The point $B_t$ is on the line segment $AB$ or not. If $B_t$ is on the line segment $AB$ then the tail moved left and also $B_t$ is on the line segment $AB'$. From Proposition~\ref{prop2}, $\lambda_{B'}'\prec \lambda_A'$. This gives that $\mathcal{C}'$ is asymmetric and the tail robot remains tail in $\mathcal{C}'$. For the remaining case, we have $B_t$ is not on the line segment $AB$. Then the tail moved right and also $B_t$ is not on the line segment $AB'$. Since the tail was at $B$ in $\mathcal{C}$, so if $\lambda_{A}$ and $\lambda_{B}$ are the binary strings for $\mathcal{C}$ then $\lambda_{B}\prec \lambda_{A}$. Then from Proposition~\ref{prop2}, a similar conclusion as the former case can be made. Thus, after each movement of the tail towards $t_{target}$ the tail remains the tail until it reaches $t_{target}$, and after a finite number of movements tail reaches $t_{target}$ that completes the target formation.
    
\end{proof}

\begin{restatable}
[Statement of Theorem~\ref{th0}]{thm}{goldbach}
From any asymmetric initial configuration, the algorithm~\ref{algo:line} can form any target pattern on an infinite grid line within finite time under an asynchronous scheduler.

\end{restatable}
\begin{proof}
    Let $\mathcal{C}$ be an asymmetric initial configuration. If for the initial configuration, $\mathcal{C}'=\mathcal{C}'_{target}$ is true, then from Lemma~\ref{tail2} after a finite time target pattern gets formed. Suppose, initial configuration satisfy $\mathcal{C}'\ne\mathcal{C}'_{target}$, there can be two exhaustive cases. Either the tail is at the left side of the $t_{target}$ or not. If the tail is at the left side of the $t_{target}$, then from Lemma~\ref{tail1} we have that after finite movements, the tail reaches the $t_{target}$. Thus after a finite time, we arrive at the configuration where $\mathcal{C}'\ne\mathcal{C}'_{target}$ is true and the tail is either on the $t_{target}$ or on the right side of $t_{target}$. Then from Lemma~\ref{l3} we have that, after finite time $\mathcal{C}'=\mathcal{C}'_{target}$ becomes true. Now at this time, if the tail is on the $t_{target}$ then the target pattern has been formed. If the tail is at the right side of the $t_{target}$, then from Lemma~\ref{tail2}, after a finite time target pattern gets formed. The flow of the algorithm is given in Fig.~\ref{fig:flow0}.
\end{proof}

\begin{figure}[ht!]
    \centering
    \includegraphics[width=.99\linewidth]{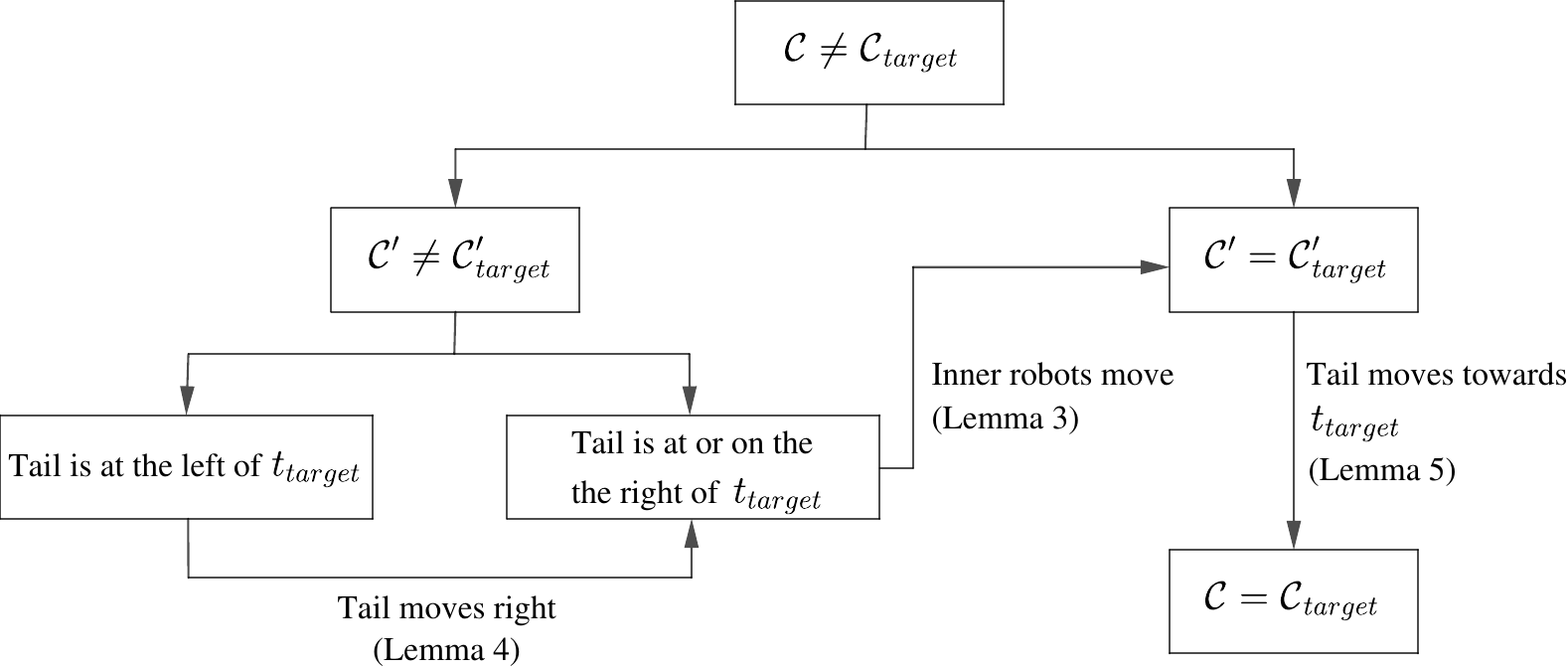}
    \caption{\footnotesize Flow of the Algorithm \textsc{ApfLine}}
    \label{fig:flow0}
\end{figure}

\section{Correctness of the Proposed Algorithm on Rectangular Grid}
\subsection{Correctness of Phase~I}\label{ph1proof}

\begin{theorem}\label{ph1}
 If at some time $t$, we have an asymmetric configuration $\mathcal{C}$ in phase~I, then after one move of the tail upwards, the tail robot remains the tail and $\neg(C_1\land C_3)$ remains true. Also, after a finite number of moves by the tail upwards $C_4\land C_5\land C_6$ becomes true.
\end{theorem}
\begin{proof}
    Let $ABCD$ be the SER of the configuration at time $t$ and after a movement of the tail, say $r_t$, upwards the configuration becomes $\mathcal{C}'$. Let $ABC'D'$ be the SER of $\mathcal{C}'$, then $|AB|<|AD'|$ and $r_t$ is the only robot on $C'D'$ line-segment. Let $A$ be the leading corner of the $\mathcal{C}$. Note that movement of $r_t$ is such that neither $C'$ nor $D'$ can be the leading corner of $\mathcal{C}'$. If $C_{10}$ is true, then $B$ is the leading corner of $\mathcal{C}'$. Otherwise, $A$ remains the leading corner of $\mathcal{C}'$. In both cases, configuration remains asymmetric and the tail remains the tail. If either of $C_1$ or $C_3$ is false in $\mathcal{C}$, then that remains false in $\mathcal{C}'$. Now it is easy to observe that after a finite number of movements of the robot $r_t$ upwards $C_4\land C_5\land C_6$ becomes true. 
\end{proof}

\subsection{Correctness of Phase~II}\label{ph2proof}

\begin{theorem}\label{ph2}
 If at some time $t$, we have an asymmetric configuration $\mathcal{C}$ in phase~II, then after one move of the head towards left the new configuration remains asymmetric and the coordinate system remains unchanged. Also, after a finite number of moves by the head, $C_8$ becomes true, and phase~II terminates with $(C_4\land C_5\land C_6\land C_8) \land (\neg C_2\lor(\neg C_3\land C_2))$ true.
\end{theorem}
\begin{proof}
    Let $ABCD$ be the SER of $\mathcal{C}$ such that $|AD|\ge|AB|$ and $\lambda_{AB}$ be the lexicographically largest string. Let $i^{th}$ term of $\lambda_{AB}$ be the first nonzero term. After the movement of the head towards left, if $\lambda^{new}_{AB}$ is the updated string, then $(i-1)^{th}$ term of $\lambda^{new}_{AB}$ be the first nonzero term. And the first $(i-1)$ terms of the rest of the other considered strings are zero. Therefore, $\lambda^{new}_{AB}$ is the strictly largest string after the movement of the head towards the left. Hence the new configuration remains asymmetric and the coordinates system does not change. After a finite number of movements of the head towards the left, it reaches the origin that results $C_8=$ true.  
\end{proof}

\subsection{Correctness of Phase~III}\label{ph3proof}

\begin{theorem}\label{ph3}
 If at some time $t$ we have an asymmetric configuration $\mathcal{C}$ in phase~III such that $C_{10}$ is not true, then after one move of the tail the configuration remains asymmetric and the coordinate system remains unchanged. Also, after one move by the tail towards the right, we still have $\neg C_2\land C_4\land C_8=$ true. After a finite number of moves by the tail $C_4\land C_5\land C_6 \land C_7 \land C_8 \land C_9 \land \neg C_2 $ becomes true.
\end{theorem}
\begin{proof}
    Let $ABCD$ be the SER of the $\mathcal{C}$ with $\lambda_{AB}$ as the lexicographically largest string. If $C_9$ is false then after one movement of the tail upwards $C_9$ becomes true and finally, the configuration satisfies $C_4\land C_5\land C_6 \land C_8 \land C_9 \land \neg C_2 \land \neg C_7=$ true. Suppose $C_9$ is true. Let after one movement of the tail the SER remains the same. Then $C_6$ remains true and so we need to compare four strings. Note that the $D$ node is empty. Since $A$ node is occupied (because $C_8$ is true), so $\lambda^{new}_{AB}$ remains larger than $\lambda^{new}_{DC}$. Since $C_{10}$ is false, so by the movement of tail robot $\lambda^{new}_{AB}$ remains larger than $\lambda^{new}_{BA}$. So we need to compare only $\lambda^{new}_{AB}$ and $\lambda^{new}_{CD}$. It is easy to see that after the movement of the tail $\lambda^{new}_{AB}$ remains the largest string. Thus, the configuration remains asymmetric and the coordinate system remains unchanged. it is easy to observe that after the movement in this case $\neg C_2\land C_4\land C_5\land C_6\land C_8 \land C_9$ remains true. 
    
    Next, suppose after one movement of the tail the SER gets changed. This can occur in two ways: either when $m=n+1$ and the tail moves upward or when the tail is at $C$ and it moves rightwards. Suppose $m=n+1$ and the tail moves upward. Let the SER become $ABC'D'$ after the move. In Theorem~\ref{ph1}, it is shown that if $C_{10}$ is false and the tail moves upward then the configuration remains asymmetric and the coordinate system does not change. Suppose the tail is at $C$ and it moves rightwards. Let the new SER be $AB'C'D$. Then note that $B'$ and $D$ are empty but $A$ and $C'$ are occupied. Since before the move of the tail $m>n+1$, $AB'C'D$ is a non-square rectangle. Then we have to compare only two strings $\lambda^{new}_{AB'}$ and $\lambda^{new}_{C'D}$. It is easy to see that $\lambda^{new}_{AB'}$ remains the largest string. Thus, the configuration remains asymmetric, and the coordinate system remains unchanged. It is easy to show that $C_4\land C_6\land C_8$ remains true for both the movements. For both types of movements, $C_4 \land C_6\land C_8\land C_9$ remains true. If $C_5$ becomes false during the upward movement, the algorithm enters into Phase~I. And the tail moves one hop upward to make $C_5$ true. After that, the algorithm will again enter into Phase~III. Thus, after a finite number of movements of the tail robot $C_7$ becomes true.    
      
\end{proof}

\begin{theorem}\label{ph30}
 If at some time $t$ we have an asymmetric configuration $\mathcal{C}$ in phase~III such that $C_{10}$ is true, then after a finite number of moves by the tail $C_4\land C_5\land C_6 \land C_7 \land C_8 \land C_9 \land \neg C_2 $ becomes true. The configuration remains asymmetric during the movement of the tail.
\end{theorem}
\begin{proof}
    Let $ABCD$ be the SER of the $\mathcal{C}$ with $\lambda_{AB}$ be the lexicographically largest string. Let $L$ be the line of symmetry of $\mathcal{C}$. Since $\mathcal{C}$ is asymmetric, the head and the tail must be on the same side of the $L$. Suppose that after one movement of the tail towards the left the SER remains the same. Since the height of the configuration is odd, $\lambda_{AB}^{new}$ becomes larger than before whereas $\lambda_{BA}^{new}$ becomes smaller than before. So, if after the move the tail does not reach $D$, then $\lambda_{AB}$ remains the largest string. Suppose, after the movement, the tail reaches $D$. Because $C$ is empty node, so $\lambda_{DC}$ is smaller than both $\lambda_{AB}$ and $\lambda_{BA}$. Thus, if the SER remains the same after the movement of the tail then $\lambda_{AB}$ remains the largest string. So the configuration remains asymmetric and the coordinate system does not change. 
    
    Suppose after the movement of the tail the SER changes. If the tail's movement upward is responsible for the change of the SER then using the same argument of Theorem~\ref{ph3} the configuration remains asymmetric and the coordinate system does not change. Suppose the tail's movement towards the left is responsible for the change of the SER. Then according to Phase~III, the SER is still a non-square rectangle. Let the new configuration be $A'BCD'$. Note that only, $B$ and $D'$ are occupied corners. Clearly, $\lambda_{BA'}$ is the larger one. Thus, the configuration remains asymmetric but the coordinate system changes. After this movement $C_{10}$ becomes false. Thus, after a finite number of moves towards left it reaches at the corner $D$. Then, the algorithm might enter Phase~I and after a finite move upwards the algorithm again enters into Phase~III. Then after one movement of the tail towards the left makes $C_{10}$ false so from Theorem~\ref{ph3}, after a finite number of moves by the tail $C_4\land C_5\land C_6 \land C_7 \land C_8 \land C_9 \land \neg C_2 $ becomes true. 
\end{proof}

\subsection{Correctness of Phase~IV}\label{ph4proof}

\begin{theorem}\label{ph4}
 If at some time $t$, we have an asymmetric configuration $\mathcal{C}$ in phase~IV, then after any movement of inner robots according to the function \texttt{rearrange}, the new configuration remains still asymmetric and the coordinate system remains unchanged. During these movements of inner robots, $C_4\land C_5\land C_6\land C_7 \land C_8$ remains true. After a finite number of movements of inner robots according to the function \texttt{rearrange} $C_2$ becomes true.
\end{theorem}
\begin{proof}
    Let $ABCD$ be the SER of $\mathcal{C}$ with $\lambda_{AB}$ as the largest string. First, we show that no inner robot moves to the $BC$ line. On the contrary, Suppose an inner robot $r_i$ lands at a node $v_2$ on the $BC$ line from a point $v_1$ at the left side of $v_2$. There are two cases: $v_1$ is at the \texttt{left} of $v_2$ or $v_1$ is at the \texttt{right} of $v_2$. Suppose, $v_1$ is at the \texttt{left} of $v_2$. Then consider the downwards nodes $u_1$ and $u_2$ of $v_1$ and $v_2$ respectively. $t_i$ cannot be at $u_2$ or $v_2$, because there is no target position on $BC$ line segment. If $t_i$ is at $u_1$ then according to \texttt{rearrange} $r_i$ shall move to $u_1$. Otherwise, $t_i$ can be at the \texttt{left} of the $u_1$. In that case, if $u_1$ is occupied by a robot then there is a robot in between $r_i$ and $t_i$ in the path $\mathcal{P}$. In that case, $r_i$ does not move according to \texttt{rearrange}. Otherwise, $u_1$ is unoccupied, then according to \texttt{rearrange} $r_i$ shall move to $u_1$. Thus, in all the cases, $r_i$ does not move to $v_2$. Similarly, if $v_1$ is at the \texttt{right} of $v_2$, then also we can show that that $r_i$ does not move to $v_2$. Next, since $C_4$ is true, there is no target position of the line segment $DC$ or above it. For this reason, we can show that no inner robot moves to $CD$ line segment according to \texttt{rearrange}.

    Suppose $r_i$ is an inner robot. Note that in $\mathcal{C}$, $A$ and $C$ corners of SER are occupied and others are not. Movements of inner robots are designed such that SER does not change with that. Because $C_7$ ($C_4$) is true, there is not any other robot except the tail or any target position on the $BC$ ($CD$) line. We showed that no inner robot ever moves on to $BC$ line segment in function \texttt{rearrange}. So, $B$ remains unoccupied after the movement of $r_i$. Next, we showed that no inner robot moves to the $CD$ line segment, so $D$ remains unoccupied after the movement of $r_i$. Thus, it is sufficient to consider only two strings $\lambda_{AB}$ and $\lambda_{CD}$. 
    
    Note that, strings $\lambda_{AB}$ and $\lambda_{CD}$ are the two binary strings of the path $\mathcal{P}$ from the different ends. Now, the movements of the inner robots are designed in such a way that if $t_i$ is at the \texttt{left} (\texttt{right}) of $r_i$ then $r_i$ moves to a point that is at the \texttt{left} (\texttt{right}) of the $r_i$. This is an adoption of the Algorithm~\ref{algo:line}. Thus, from Lemma~\ref{l1} and Lemma~\ref{l2}, it is evident that during the movement of the inner robots, $\lambda_{AB}$ remains the larger string. Thus, the configuration remains asymmetric and the coordinate system does not change. Since the head and the tail do not move at all, all inner robot moves inside the SER formed by the head and the tail, and no inner robot moves onto line-segment $BC$ or $CD$, so $C_4\land C_5\land C_6\land C_7 \land C_8$ remains true throughout the movements of the inner robots. Also from Lemma~\ref{l3}, after a finite time, all inner robots take their respective positions making $C_2$ true.
\end{proof}

\subsection{Correctness of Phase~V}\label{ph5proof}

\begin{theorem}\label{ph5}
 If at some time $t$, we have an asymmetric configuration $\mathcal{C}$ in phase~V, then after a finite number of movements of the tail $C_2 \land C_3\land C_4\land C_5\land C_6 \land C_8$ becomes true. If at this point the configuration has vertical symmetry then $C_1$ must be true.
\end{theorem}
\begin{proof}
    Let $ABCD$ be the SER of $\mathcal{C}$ with $\lambda_{AB}$ as the largest string. Let $AB'C'D'$ be the SER of the $\mathcal{C}'$. Suppose the tail is at $T$. Suppose $C_{10}$ is not true. Since $C_4$ is true, for this case, throughout the movement of the tail towards $T'$ the configuration remains asymmetric and the coordinate system remains unchanged. After a finite number of movements, $C_3$ becomes true. Next, suppose $C_{10}$ is true. In this case, if $T$ and $T'$ both are on the $DE$ line segment then it is easy to see that the coordinate system remains invariant during the movement of the tail. Also, if both $T$ and $T'$ are at the right side of the $C''$, then the tail remains the tail while the movement of the tail towards $T'$.

    Next, suppose $T$ is on the $DE$ line segment but $T'$ is at the right of the $C''$. In this case, the tail moves leftward. When the tail moves at the left of the $D$, then the coordinate system flips. The tail remains the tail but the robot at $B'$ becomes the head. Then it reduces it to one of the previous cases. Next, suppose $T'$ is on the $DE$ line segment but $T$ is at the right of the $C''$. In this case, the tail moves leftwards towards $T'$. When the tail reaches $C''$, the coordinate system flips and the robot at $B'$ becomes the head. Thus, it again reduces to one of the previous cases. Therefore, $C_3$ becomes true after a finite number of movements of the tail.

    After $C_3$ becomes true, if the configuration has a vertical symmetry then the $T'$ must be $E$ which coincides with a grid node. Before the tail moved on $E$ the configuration was asymmetric. Without loss of generality let $A$ be the leading corner. Since $C_8$ was true in this phase, $A$ was occupied by the head robot. So after $C_3$ becomes true, $A$ is occupied. Due to the symmetry, $B'$ must be occupied by a robot. According to the embedding of the target pattern, both the left and the right bottom corners also have target positions. If possible let both $A$ and $B'$ are not unoccupied by any target position. This gives, at this point three robots three robots are not at their target position, which contradicts the assumption $C_2=$ true according to which $k-2$ inner robots are at their target positions.
    Since $A$ and $B'$ both are occupied by the robots, the $h_{target}$ must be occupied. Thus, $C_1$ is true.
\end{proof}

\section{Correctness of Phase~VI}\label{ph6proof}

\begin{theorem}\label{ph6}
    If at some time $t$, we have an asymmetric configuration in phase~VI, then after a finite number of movements by the head towards $h_{target}$, $C_1\land C_3\land C_4\land C_5\land C_6$ becomes true. If the configuration becomes such that it has a vertical symmetry then $C_{target}$ has the same.
\end{theorem}
\begin{proof}
    Suppose $ABCD$ is the SER of the current configuration $\mathcal{C}$ with $\lambda_{AB}$ as the largest string. Suppose the head robot is at $H$ in $\mathcal{C}$ and $h_{target}$ is at $H'$. First, suppose the $H'$ is at the left of the $H$. Then following from the Theorem~\ref{ph2}, when the head moves towards the left, the configuration remains asymmetric and the coordinates system remains the same. Suppose the $H'$ is at the right of the $H$. Suppose the tail is at $T$ and the $t_{target}$ is at $T'$. According to the target embedding, the SER of the embedded target pattern should also be $ABC'D'$ with $\lambda^{target}_{AB}$ as a lexicographically largest string, where $T'$ is on the line segment $C'D'$. Since $C_3$ is true, $T$ and $T'$ are at the same vertical line. If the tail moves from $T$ to $T'$ and the head reaches from $H$ to $H'$, then the target is formed. Let $\mathcal{C}_h$ be the configuration if the head moves from $H$ to $H'$. Then $\lambda^{target}_{AB}$ will be largest string in $\mathcal{C}_h$. $\lambda^{target}_{AB}$ may not be the strictly largest string in $\mathcal{C}_h$, only when $C_{target}$ has a vertical symmetry. In both cases, while the movement of the head the configuration remains asymmetric and the coordinate system does not change. Therefore, after a finite number of moves by the head $C_1$ becomes true. And, throughout the movements of the head, $C_3\land C_4\land C_5\land C_6$ remains true.
\end{proof}

\section{Correctness of Phase~VII}\label{ph7proof}

\begin{theorem}\label{ph7}
     If at some time $t$ we have a configuration $\mathcal{C}$ within phase~VII, then after a finite number of movements of the tail towards $t_{target}$, $C_0$ becomes true. 
\end{theorem}
\begin{proof}
    Suppose $ABCD$ is the SER of the current configuration $\mathcal{C}$. Let $\mathcal{C}$ be asymmetric. For such a case, let $\lambda_{AB}$ be the lexicographically largest string. Suppose the tail is at $T$ and $t_{target}$ is at $T'$. If $T'$ is above $T$, then there can be two cases: $C_{10}$ is true or not. If $C_{10}$ is true, then in Theorem~\ref{ph1} it is shown that when the tail moves upward the coordinate flips, but the tail remains the tail. Even if the coordinate system flips, due to the symmetry of $\mathcal{C}'$, $C_1$ remains true. So after a finite number of movements of the tail upwards, the tail reaches $T'$ resulting in $C_0=$ true. Next suppose $C_{10}$ not true. Then from Theorem~\ref{ph1}, it is shown that when the tail moves upward, the configuration remains asymmetric and the coordinate system remains unchanged. Therefore after a finite number of movements of the tail upwards, the tail reaches $T'$ resulting in $C_0=$ true. Next suppose, $T'$ is below $T$. Let $ABC'D'$ be the SER of the $\mathcal{C}_{target}$, then $T'$ is on the $C'D'$ and $\lambda_{AB}$ is the largest string in $\mathcal{C}_{target}$. Since $C_1$ is true so, all the target positions are occupied but $t_{target}$. Thus, it is easy to see that if $t_{target}$ moved upwards then $\lambda_{AB}$ becomes the strictly largest string in new $\mathcal{C}_{target}$. Thus, when the tail moves from $T$ to $T'$, the $\lambda_{AB}$ remains the strictly largest string.

    Next, let $\mathcal{C}$ be symmetric. Since the initial configuration is asymmetric, so the symmetric configuration can be formed via Phase~V and Phase~VI. Thus the symmetry is a vertical symmetry, resulting in $\lambda_{AB}=\lambda_{BA}$ as the strictly largest string. For both the Phase~V and Phase~VI, it terminates with $C_4=$ true. Therefore, the $T'$ is at the downward of the tail. Using the similar argument as above, $\lambda_{AB}=\lambda_{BA}$ remains the strictly largest string after each movement of the tail downwards until it reaches $T'$.   
\end{proof}


\begin{lemma}\label{phase0}
     Any asymmetric initial configuration, satisfying $C_0=$ false, falls under one and exactly one of the seven phases of the proposed algorithm. 
\end{lemma}
\begin{proof}
From Fig~\ref{fig:main_flow} the proof follows.
    \begin{figure}[h!]
        \centering
        \includegraphics[width=1\linewidth]{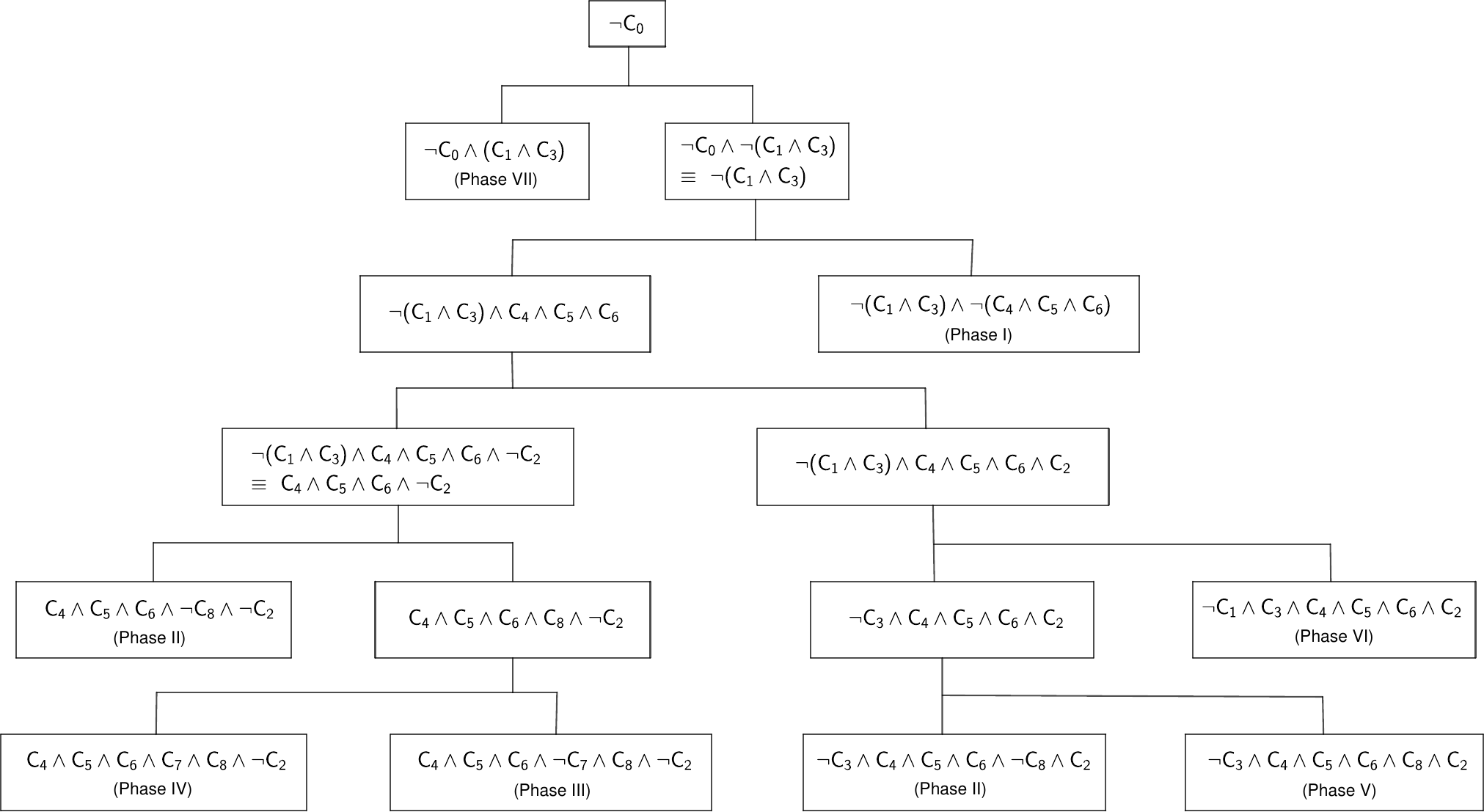}
        \caption{For any configuration with $C_0=$ false belongs to one of the seven phases}
        \label{fig:main_flow}
    \end{figure}
\end{proof}

\subsection{Proof of Theorem~\ref{main}}\label{mainproof}

\begin{restatable}
[Statement of Theorem~\ref{main}]{thm}{goldbach}
     The proposed algorithm can form any pattern consisting of $k$ points by a set of $k$ oblivious asynchronous robots if the initial configuration formed by the robots is an asymmetric configuration and has no multiplicity point. 
\end{restatable}
\begin{proof}
    Let $\mathcal{C}_i$ be an initial asymmetric configuration formed by $k$ robots with no multiplicity points. Let $\mathcal{C}_{target}$ be any target configuration consisting of $k$ target positions. According to the Lemma~\ref{phase0}, if $\mathcal{C}_i\ne \mathcal{C}_{target}$ then the algorithm starts from any of the seven phases. Suppose after a finite time the algorithm is in one of the seven phases, then we show that after finite time $C_0$ becomes true. If at some time the algorithm is in some specific phase, then next which phase the algorithm can enter. In Fig.~\ref{fig:flow}, a digraph is given that shows the phase transitions. This digraph can be created from the support of Theorem~\ref{ph1}, \ref{ph2}, \ref{ph3}, \ref{ph30}, \ref{ph4}, \ref{ph5}, \ref{ph6}, \ref{ph7}. The only cycle in the digraph is the cycle induced by phase~II and phase~III. From Theorem~\ref{ph3} and Theorem~\ref{ph30}, we can conclude it does not create any live-lock there. From the diagram, we can conclude that any path starting from any of the phases leads to phase~VII after finite time. From Theorem~\ref{ph7}, the phase~VII results $C_0=$ true within finite time.  
    \begin{figure}[h!]
        \centering
        \includegraphics[width=.6\linewidth]{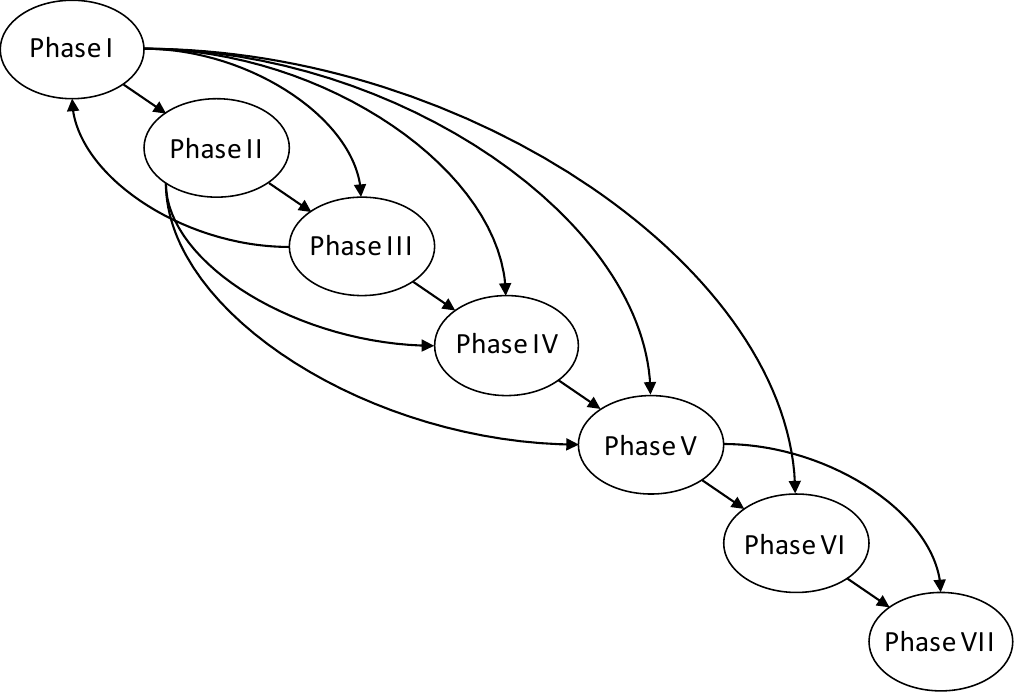}
        \caption{Phase transition digraph}
        \label{fig:flow}
    \end{figure}
\end{proof}

\subsection{Proof of Theorem~\ref{space}}\label{spaceproof}

\begin{restatable}
[Statement of Theorem~\ref{space}]{thm}{goldbach}
     Let $\mathcal{D}=\max\{m,n,m',n'\}$ where $m\times n$ $(m\ge n)$ is the dimension of the SER of the initial configuration and $m'\times n'$ $(m'\ge n')$ is the dimension of the SER of the target configuration. Then space complexity of the proposed algorithm is at most $\mathcal{D}+4$. More precisely, if $M=\max\{m,m'\}$ and $N=\max\{n,n'\}$, then the proposed algorithm takes the space enclosed by a rectangle of dimension $(M+4)\times (N+1)$. 
\end{restatable}

    \begin{proof}
    The tail robot is the only robot that moves out of the current SER. The head robot only moves in Phase~II and Phase~VI and the head robot either moves to the corner or the current SER moves towards $h_{target}$. In Phase~IV the inner robots move but they do not move outside the current SER throughout the function~\texttt{Rearrange}. In the rest of the Phases~I, III, V, VII, the tail robot moves only. The tail robot is responsible for expanding the size of the configurations. Consider the rectangle $\mathcal{R}$ of dimension $\mathcal{D}\times\mathcal{D}$ that contains the initial configuration and the embedded target pattern. Next, let us define the size of a still configuration $\mathcal{C}$ is the dimension of the smallest enclosing square that contains all the robots in $\mathcal{C}$.
    
    In Phase~I, the tail moves upward. If the initial configuration already satisfies $C_4\land C_6\land \neg C_5=$ true then in order to make $C_5$ true, the tail has to move one hop upwards, outside the $\mathcal{R}$. Thus, the size of the current SER becomes $\mathcal{D}$+1. If for the initial configuration, $C_4$ is true but $C_6$ is not true. Then by one movement upwards $C_6$ becomes true. If after this $C_5$ is not true, then one more movement upwards by tail makes $C_5$ true. So, finally, the size of the SER is $\mathcal{D}$+2. If $C_4$ is not true for the initial configuration, then the tail moves upward until it reaches the horizontal line $t_{target}$. After this, the tail moves one step upwards to make $C_4$ true, and $C_6$ becomes true with this move. If at this point $C_5$ is not true then it again moves upwards one hop. Thus, the size of the SER becomes $\mathcal{D}$+2. 

    Next, suppose the algorithm enters Phase~III from Phase~I. Suppose $C_{10}$ is false. If $C_9$ is not true then the tail moves upwards and $C_{9}$ becomes true. After that, only $C_5$ will become false and the algorithm enters into Phase~I again. After one move of the tail upwards in Phase~I, $C_5$ becomes true. The algorithm again enters in Phase~III with $C_{9}=$ true. At this point $m>n+1$, because when the first time comes from Phase~I, $C_6=$ true assures $m\ge n$ and while the second time entering Phase~I, the upward movement of tail assures $m>n+1$. Thus, more two upward movements by the tail take place making the size of the SER $\mathcal{D}+4$. Next if $C_{10}$ is true, then on being $C_9=$ false the tail moves upward. Again for the same reason as the last case finally the size of the SER becomes $\mathcal{D}+4$ when all the upward movements of the tail are done. So by upward movement the larger dimension of the SER containing the current configuration becomes $M+4$ at most. Now in order to make $C_7$ true the tail robot at most needs to step away one hop rightwards or leftwards form $\mathcal{R}$. So after making $C_7$ true the size of the SER remains at most $\mathcal{D}+4$ and the smaller dimension of the SER becomes $N+1$ at most. Now, one can easily verify from the above argument that if the initial configuration is such that the algorithm first enters Phase~III, then it terminates with making the size of the SER at most $\mathcal{D}+2$ and the SER remains inside a rectangle of dimension $(M+2)\times (N+1)$.

    In Phase~V the tail moves only on the horizontal line $H_t$ that it contains and moves towards the $T'$ on the $H_t$ such that $T'$ and $t_{target}$ are on the same vertical line. If at the beginning of this phase, the tail is inside the $\mathcal{R}$ then it does not step out of it. If the tail is outside the $\mathcal{R}$ then the horizontal movements of the tail do not increase the size of the current SER. Thus movements of the tail in this phase do not consume extra space.

    In Phase~VII, the tail moves towards the $t_{target}$ vertically. It is easy to see that the movements of this phase also do not consume any extra space.

    Thus, all the robots move inside a $(M+4)\times(N+1)$ dimensional rectangle. Therefore total space consumed by the proposed algorithm is at most $\mathcal{D}+4$. 
    
    \end{proof}

    \subsection{Proof of Theorem~\ref{th:move}}\label{moveproof}

\begin{restatable}
[Statement of Theorem~\ref{th:move}]{thm}{goldbach}
     The proposed algorithm requires each robot to make $O(\mathcal{D})$ movements, hence the move-complexity of the proposed algorithm is $O(k\mathcal{D})$.
\end{restatable}
    \begin{proof}

    First, consider the movement of the head robot. When the head robot moves towards the origin in Phase~II, from Theorem~\ref{ph2} it remains the head throughout the movement. It is easy to see that, for this case total number of moves is almost $\mathcal{D}/2$. In Phase~VI, when the head moves towards $h_{target}$, the head remains the head. In this phase at most $\mathcal{D}/2$ moves by the head robot is required. Thus total number of moves by the head robot is at most $\mathcal{D}$. 
    
    Next, consider the tail robot movements. In Phase~I, the tail robot needs to move upward at most $\mathcal{D}$ steps to make $C_4$ true. Further, to make $C_5$ and $C_6$ true, the tail needs to move at most two steps upward. In Phase~III, the tail needs to move horizontally at most $\mathcal{D}+1$ steps. In Phase~V, again the tail robot moves horizontally at most $\mathcal{D}+1$ steps. In Phase~VII, the tail moves vertically at most $\mathcal{D}+k_0$ steps, where $k_0$ is a constant less than $\mathcal{D}$. Thus, the tail needs to move $O(\mathcal{D})$ steps in total. 
    
    Next, consider the movements of the inner robots. Suppose $r_i$ is an inner robot located at point $R_i$ in the initial configuration. Let $t_i$ be the corresponding target position of $r_i$. If $t_i$ and $R_i$ are on the same horizontal line, then after that $r_i$ moves horizontally towards $t_i$ and after reaching $t_i$ never moves afterward. Otherwise, suppose $t_i$ is at the \texttt{left} of the $R_i$. Then $t_i$ is on a horizontal line downwards to that of $R_i$. Until the downwards node of $r_i$ is on the horizontal line that contains $t_i$, the downward node of the $r_i$ will always be at the \texttt{right} of $t_i$. So, $r_i$ moves downwards through the vertical line containing $R_i$. Suppose $r_i$ and $t_i$ are on a neighboring horizontal line. There are two cases: the downward node of $r_i$ at the \texttt{left} of the $t_i$ or at the \texttt{right} of the $t_i$. If the downward node of $r_i$ at the \texttt{left} of the $t_i$, then the $r_i$ moves downwards, and $r_i$ and $t_i$ are on the same horizontal line. After that $r_i$ moves towards $t_i$ through the horizontal line that contains both. Thus, it is easy to see that the path traveled by the $r_i$ to reach $t_i$ has a minimum length which is at most $2\mathcal{D}$. Next, suppose the downward node of $r_i$ at the \texttt{right} of the $t_i$. Then the $r_i$ moves to its \texttt{left}. It keeps moving towards its \texttt{left} until $r_i$ and $t_i$ are on the same vertical line. Then after one movement downwards $r_i$ reaches $t_i$. In this case, the path traveled by the $r_i$ to reach $t_i$ has a minimum length which is at most $2\mathcal{D}$. Thus, in either case, the $r_i$ moves at most $2\mathcal{D}$ steps.

    If $t_i$ is at the \texttt{right} of the $R_i$, then one can similarly show that the path traversed by $r_i$ from $R_i$ to $t_i$ has length at most $2\mathcal{D}$. Thus each robot makes $O(\mathcal{D})$ moves throughout the execution of the algorithm. Hence the result follows.
    \end{proof}

\end{document}